\documentclass[
a4paper,
reprint,
amsmath,
amssymb,
superscriptaddress,
onecolumn,
accepted=2026-07-15]
{quantumarticle}
\pdfoutput=1
\usepackage[numbers,sort&compress,square]{natbib}

\usepackage{graphicx}
\usepackage{dcolumn}
\usepackage{bm}
\usepackage[colorlinks]{hyperref}
\usepackage{mathrsfs}
\usepackage{braket}
\usepackage{mathtools}
\usepackage{autobreak}
\usepackage{algorithmic}
\usepackage{algorithm}
\usepackage{color}
\usepackage{amsthm}
\usepackage{quantikz}
\usepackage{empheq}
\usepackage[caption=false]{subfig}
\usepackage{booktabs, multirow}
\allowdisplaybreaks

\newtheorem{thm}{Theorem}

\newtheorem{dfn}{Definition}
\newtheorem{lem}{Lemma}
\newtheorem{cor}{Corollary}

\newtheorem{problem}{Problem}

\def\theenumi{\roman{enumi}}

\begin{document}

\title{On the quantum computational complexity of classical linear dynamics with geometrically local interactions: Dequantization and universality}

\author{Kazuki Sakamoto}
\email{kazuki.sakamoto.osaka@gmail.com}
\affiliation{%
Graduate School of Engineering Science, The University of Osaka\\
1-3 Machikaneyama, Toyonaka, Osaka 560-8531, Japan.
}

\author{Keisuke Fujii}
\affiliation{%
Graduate School of Engineering Science, The University of Osaka\\
1-3 Machikaneyama, Toyonaka, Osaka 560-8531, Japan.
}
\affiliation{%
Center for Quantum Information and Quantum Biology, The University of Osaka 560-0043, Japan.
}
\affiliation{
Center for Quantum Computing, RIKEN, Hirosawa 2-1, Wako Saitama 351-0198, Japan.
}

\begin{abstract}
{
The simulation of large-scale classical systems in exponentially small space on quantum computers has gained attention.
The prior work demonstrated that a quantum algorithm offers an exponential speedup over any classical algorithm in simulating classical dynamics with long-range interactions.
However, many real-world classical systems, such as those arising from partial differential equations, exhibit only local interactions. 
The question remains whether quantum algorithms can still provide exponential speedup under this condition.
In this work, we thoroughly characterize the computational complexity of simulating such geometrically local systems on quantum computers.
First, we dequantize the quantum algorithm for simulating short-time (polynomial-time) dynamics of such systems. 
This implies that the problem of simulating this dynamics does not yield any exponential quantum advantage. 
Second, we show that simulating short-time dynamics is at least as hard as polynomial-time and linear-space probabilistic classical computation.
Third, we show that the computational complexity of simulating long-time (exponential-time) dynamics is captured by exponential-time and polynomial-space quantum computation.
This suggests a super-polynomial time advantage when restricting the computation to polynomial-space, or an exponential space advantage otherwise.
This work offers new insights into the complexity of classical dynamics governed by partial differential equations, providing a pathway for achieving quantum advantage in practical problems.
}
\end{abstract}

\maketitle


\section{Introduction}\label{sec:introduction}
The simulation of quantum systems is one of the key problems expected to demonstrate exponential quantum advantage. 
This advantage is natural because quantum computers inherently utilize the principles of quantum mechanics as Feynman pointed out~\cite{feynman1982simulating}. 
Moreover, such speedups achieved in quantum simulation is guaranteed by rigorous analysis with computational complexity theory.
That is, quantum simulation problem is $\mathrm{BQP}$-complete~\cite{feynman1986quantum, lloyd1996universal, kitaev2002classical}, which implies that 
``any'' classical algorithm cannot simulate those quantum systems efficiently under the widely-believed complexity-theoretic assumption $\mathrm{BPP}\neq \mathrm{BQP}$.

As recent works investigated, quantum computers have potential applications to not only quantum systems but also classical systems.
Many classical systems can be described by first-order linear ordinary differential equations.
Even when higher-order derivatives are involved, these equations can be converted to first-order equations by introducing auxiliary variables~\cite{costa2019quantum, babbush2023exponential}.
Therefore, quantum algorithms for solving first-order linear ordinary differential equations have been actively studied.
One major approach converts the differential equation into a linear system of equations by discretizing the time derivative and dilating the Hilbert space~\cite{berry2014high, berry2017quantum, krovi2023improved, berry2024quantum, low2024quantum} and uses quantum linear system solvers~\cite{harrow2009quantum, costa2022optimal, dalzell2024shortcut, low2024quantumlinear, morales2024quantum}.
Another approach converts the differential equation into a Schr\"{o}dinger equation~\cite{costa2019quantum, jin2024quantum, jin2023quantum} and uses Hamiltonian simulation algorithm~\cite{lloyd1996universal, low2017optimal, low2019hamiltonian, gilyen2019quantum}.
While many other algorithms use their unique approaches respectively~\cite{childs2020quantum, fang2023time, an2023linear, an2023quantum}, in any case, we can interpret most of the algorithms as a polynomial transformation of some matrices using quantum eigenvalue transformation (QEVT)~\cite{low2024quantum}.
In the case of linear partial differential equations, quantum computers can also solve the problems since discretizing a partial differential equation leads to an ordinary differential equation.
Then quantum algorithms and their circuit implementations for specific partial differential equations appeared in physics are proposed~\cite{costa2019quantum, suau2021practical, jin2024quantum, jin2023quantum, jin2024quantum2, hu2024quantum, ma2024schr, sato2024hamiltonian, schade2024quantum, sato2025quantum, bosch2024quantum}.
It should be noted that such quantum algorithms assume an efficient initial state preparation and provide efficient access to limited information extracted from the final quantum state, such as expectation values or samples, rather than unrestricted access to the full classical solution.
In this setting, quantum computers are expected to solve certain tasks associated with classical systems with exponentially smaller resources of time and space than classical computers~\cite{costa2019quantum, sato2024hamiltonian, schade2024quantum, villanyi2025exponential, li2025exponential}.
However, at present, such exponential quantum advantages are not guaranteed from complexity theoretic perspectives in most cases.

Recent work by Babbush et al. partially solved this problem~\cite{babbush2023exponential}. 
They proposed a quantum algorithm that can simulate the classical dynamics of exponentially many coupled oscillators, e.g., masses coupled by springs, in polynomial time. 
They introduced the reduction from Newton's equation to Schr\"{o}dinger equation and then showed that the Schr\"{o}dinger equation can be solved efficiently using Hamiltonian simulation algorithm under appropriate conditions such as the sparsity of interactions.
This algorithm also relies on efficient initial state preparation and provides access only to quantities that can be efficiently extracted from the final quantum state.
They also showed that such classical dynamics can simulate universal quantum computation, that is, a certain problem obtained from that dynamics is BQP-complete. This means that the dynamics cannot be simulated efficiently on classical computers under the assumption of $\mathrm{BQP}\neq \mathrm{BPP}$.
The essential part for showing BQP-completeness is that they allow long-range interactions between masses.
However, many classical systems which appear in the real world have only geometrically local interactions.
For example, the equations obtained by discretizing partial differential equations of fluid dynamics~\cite{bharadwaj2020quantum} or plasma dynamics~\cite{dodin2021applications}, e.g., 
the wave equation~\cite{costa2019quantum}, can be regarded as classical systems with spatially local interactions.
Thus it is important to analyze the computational complexity of simulating the dynamics of such geometrically local classical systems on quantum computers.

In this work, we provide a thorough characterization of the computational complexity of simulating classical linear dynamics with geometrically local interactions on quantum computers.
Specifically, we consider classical systems of size $N=2^n$ governed by first-order linear differential equations that can be solved by $O(n)$-qubit quantum algorithms.
Note that we assume that the initial state can be efficiently sampled classically. 
This assumption ensures that quantum advantages arise solely from the time evolution, not from the description of the initial state.
We then consider the problem of either estimating a quantity given by an inner product, or sampling from the time-evolved state.
These problems are natural setups for quantum computers and corresponding quantum algorithms have an exponential space advantage over the conventional classical algorithms.
In addition, we assume that the quantum algorithm runs in $\mathrm{poly}(t,n)$-time for an evolution time $t$, which restricts our focus to classical systems that can be efficiently simulated on quantum computers.
This condition naturally holds when the original equation can be mapped to the Schr\"{o}dinger equation~\cite{babbush2023exponential}.
In contrast, it does not generally hold for arbitrary linear dynamics, thereby excluding equations that cannot be solved efficiently on quantum computers.
Throughout this work, we refer to the regime $t=\mathrm{polylog}(N)=\mathrm{poly}(n)$ as short-time evolution, and $t=\mathrm{poly}(N)=\exp{(n)}$ as long-time evolution.
We distinguish these two regimes for both computational and physical reasons. 
The short-time regime corresponds to the parameter region where quantum algorithms run efficiently, that is, in $\mathrm{poly}(n)$-time.
In contrast, the long-time regime is necessary for information to propagate across geometrically local systems and reach the boundaries, thereby enabling us to observe global and physically meaningful phenomena.
We have three main results (see Fig.~\ref{fig:summary-of-result}):
\begin{enumerate}
    \item We dequantize the quantum algorithm for simulating the classical dynamics in the short-time regime $t=\mathrm{polylog}(N)$ (see Theorem~\ref{thm:dequantize-qevt} and Theorem~\ref{thm:classical-simulation-sampling-1D-GLI}).
    That is, classical computers can solve this simulation problem with the same complexity as quantum ones up to polynomial overhead when the evolution time is small. 
    \label{result1}
    \item We show that simulating short-time ($t=\mathrm{polylog}(N)$) dynamics is at least as hard as $\mathrm{polylog}(N)$-time and $O(n)$-space probabilistic classical computation (see Sec.~\ref{subsec:1DGLI-simulate-QC-BPgate}).
    \label{result2}
    \item We show that the computational complexity of simulating long-time ($t=\mathrm{poly}(N)$) dynamics is equivalent to that of $\mathrm{poly}(N)$-time and $O(n)$-space quantum computation (see Sec.~\ref{sec:complexity-long-time-GLI}).
    This indicates that we obtain a super-polynomial time advantage in simulating long-time dynamics when restricting the computation to $\mathrm{poly}(n)$-space, and an exponential space advantage otherwise.
    A more detailed discussion is given in Sec.~\ref{sec:complexity-long-time-GLI}.
    \label{result3}
\end{enumerate}

\begin{figure}
\centerline{
\includegraphics[width=120mm, page=1]{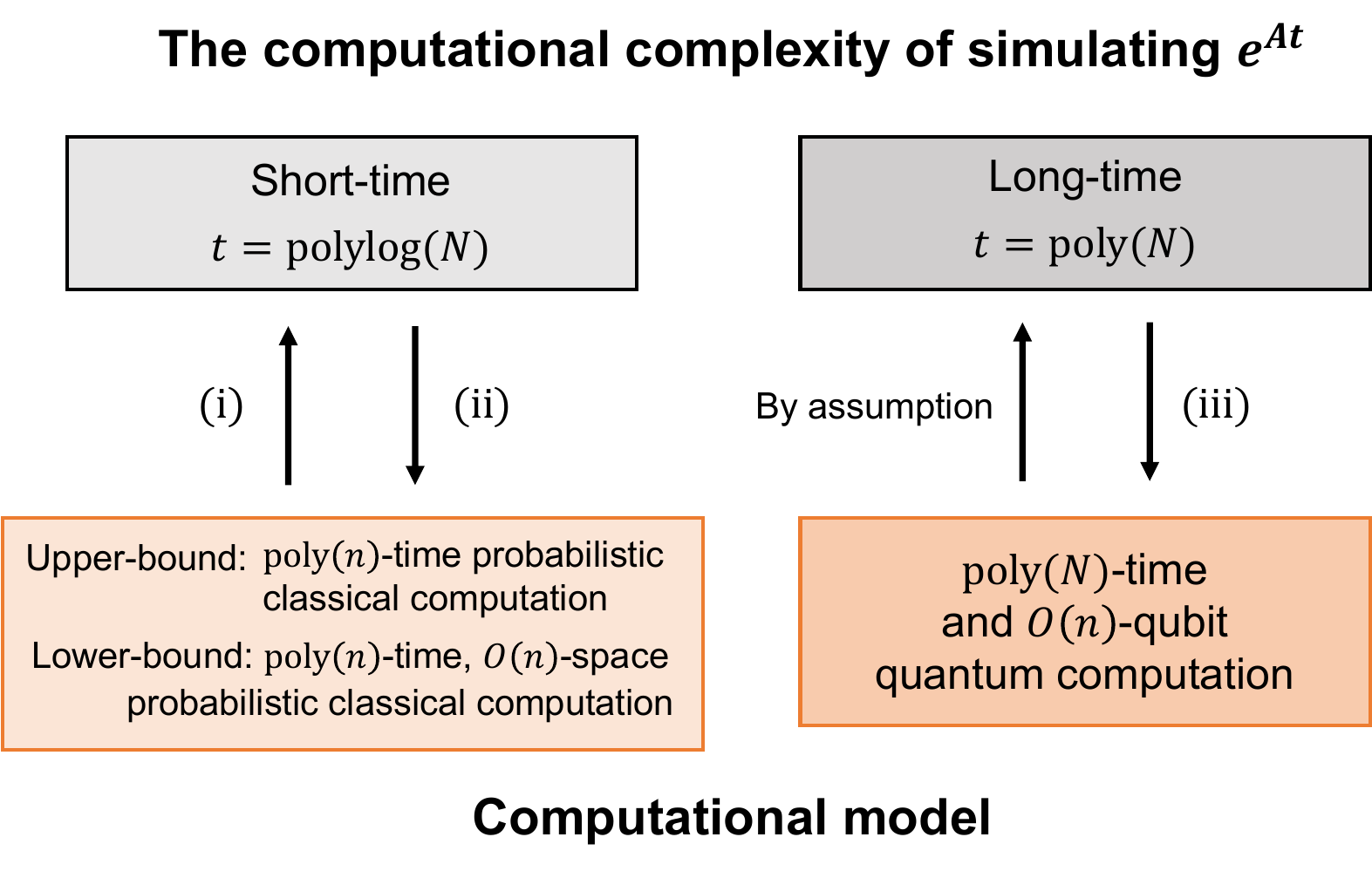}
}
\caption{
The overview of this paper.
$N=2^n$ is the system size.
Here we consider the classical linear dynamics with geometrically local interactions that can be efficiently simulated by $O(n)$-qubit quantum algorithm, i.e., the runtime of the quantum algorithm is at most $\mathrm{poly}(t,n)$.
Result (\ref{result1}) (Theorem~\ref{thm:dequantize-qevt} and Theorem~\ref{thm:classical-simulation-sampling-1D-GLI}) shows that the short-time dynamics can be simulated by $\mathrm{polylog}(N)$-time probabilistic classical computation.
In result (\ref{result2}) (Sec.~\ref{subsec:1DGLI-simulate-QC-BPgate}), we establish the opposite direction that the short-time dynamics can simulate $\mathrm{polylog}(N)$-time and $O(n)$-space probabilistic classical computation.
In result (\ref{result3}) (Sec.~\ref{sec:complexity-long-time-GLI}), we indicate that the long-time dynamics can simulate $\mathrm{poly}(N)$-time and $O(n)$-qubit quantum computation.
The opposite direction is by the assumption that the quantum algorithm runs in $\mathrm{poly}(t,n)$-time.
}
\label{fig:summary-of-result}
\end{figure}

As the first result, we provide the dequantized algorithm for simulating this classical system (Theorem~\ref{thm:dequantize-qevt}).
Dequantization is the effort to demonstrate that a quantum algorithm for some problems in the specific parameter region can be efficiently simulated on a classical computer~\cite{tang2019quantum, tang2022dequantizing}.
In the previous work, quantum singular value transformation (QSVT) of low-rank matrices~\cite{jethwani2019quantum, chia2022sampling, gall2023robust, bakshi2024improved} and general sparse matrices~\cite{gharibian2022dequantizing} are dequantized, providing new insights into quantum advantages in machine learning and quantum chemistry.
On the other hand, quantum algorithms for solving the differential equation are based on QEVT, which we consider.
To dequantize the QEVT for simulating the general classical linear dynamics with geometrically local interactions, we consider a class of geometrically local matrices.
By leveraging the locality of these matrices, we demonstrate that polynomial transformations of these matrices with polynomials of degree $d$ can be simulated classically using $\mathrm{poly}(d)$-time and $\mathrm{poly}(d)$-space.
Remarkably, these complexities are the same as those of a quantum algorithm up to polynomial overhead as long as $d=\mathrm{poly}(n)$.
This result implies that simulating the short-time ($t=\mathrm{polylog}(N)$) dynamics with geometrically local interactions does not yield any exponential quantum advantage in both time and space complexities.
Note that, if the initial states or the observables are restricted to local ones, then we can simulate the dynamics classically by considering the subsystem determined by
the ``light cone'' as discussed in Refs.~\cite{babbush2023exponential, bosch2024quantum}.
However, it was not known whether we could demonstrate the classical easiness if both the initial state and the observable are global.
We also provide an even stronger result.
Our proposed algorithm mentioned above estimates quantities which is given by an inner product.
However, estimating the inner product is not the only way to prove classical hardness, that is, sampling problems are also important.
For example, it is proven that efficient classical simulations of sampling from restricted quantum models~\cite{harrow2017quantum} are hard (implying the collapse of the polynomial hierarchy) while such models are thought to be unable to perform universal quantum computation.
Thus it is worth analyzing whether sampling from a short-time evolved state under a geometrically local matrix can be simulated classically or not.
As a result, we show that such a sampling problem can also be simulated efficiently on classical computers (Theorem~\ref{thm:classical-simulation-sampling-1D-GLI}).

The above results show the classical easiness of simulating classical linear dynamics with geometrically local interactions.
On the other hand, it is also essential to provide hardness results to capture the genuine complexity.
Our second result is that simulating the short-time ($t=\mathrm{polylog}(N)$) dynamics is at least as hard as $\mathrm{polylog}(N)$-time and $O(n)$-space probabilistic classical computation (Sec.~\ref{subsec:1DGLI-simulate-QC-BPgate}).
To show this, we embed the classical reversible circuit to some simple 1D geometrically local system of size $\approx N$, which implies that the short-time dynamics can simulate $\mathrm{polylog}(N)$-time and $O(n)$-space probabilistic classical computation.

We also analyze the long-time ($t=\mathrm{poly}(N)$) dynamics of such systems.
As the third result, we show that the computational complexity of simulating the long-time dynamics is equivalent to that of $\mathrm{poly}(N)$-time and $O(n)$-space quantum computation (Sec.~\ref{sec:complexity-long-time-GLI}).
We prove this by embedding the quantum circuit to some simple 2D geometrically local system, demonstrating that the long-time dynamics can simulate the quantum computation.
This result implies that simulating the long-time dynamics of such systems has a super-polynomial time advantage or exponential space advantage unless any $\exp(n)$-time and $\mathrm{poly}(n)$-space quantum circuit is simulated by classical computation with the same resources.
A more detailed discussion is given in Sec.~\ref{sec:complexity-long-time-GLI}.

These results indicate that a certain computational model captures the complexity of simulating the classical dynamics with geometrically local interactions and reveal how much advantage quantum computers offer for this simulation problem.
In particular, our results imply that no exponential quantum advantage exists for short-time evolution.
This is a nontrivial finding that rules out the exponential quantum speedup or space advantage in the form discussed in many prior works~\cite{costa2019quantum, sato2024hamiltonian, schade2024quantum}.
More specifically, the matrices governing the classical systems considered in these works are geometrically local matrices. 
Therefore, for short-time dynamics, our dequantized algorithm can perform the same tasks as the corresponding quantum algorithms with only polynomial overhead. 
On the other hand, for long-time evolution, we have proven an exponential space advantage under a plausible assumption.
The sharp difference between these two regimes comes from the fact that short-time geometrically local dynamics restricts information propagation to a region of radius $\mathrm{polylog}(N)$, whereas long-time evolution can propagate information to sites separated by a distance of $\mathrm{poly}(N)$ and effectively induce long-range interactions. 
We explain this point in more detail in Sec.~\ref{sec:complexity-long-time-GLI}.
To the best of our knowledge, this result provides the first complexity-theoretic evidence for an exponential space advantage~\footnote{As a related work, quantum memory advantages for stochastic-process generation have been studied in, e.g., Refs.~\cite{gu2012quantum,aghamohammadi2017extreme,korzekwa2021quantum,elliott2018superior}. These works concern stochastic-process simulation, including processes associated with classical spin chains, rather than finite-precision simulation of geometrically local classical linear dynamics considered here.} in simulating realistic classical systems with geometrically local interactions, such as systems arising from discretized partial differential equations, or a super-polynomial time advantage when the space complexity is restricted to $\mathrm{poly}(n)$.
This work provides a new insight into the complexity of classical dynamics determined by partial differential equations and deepens our understanding of the necessary elements to obtain quantum advantages.

This paper is organized as follows. 
In Sec.~\ref{sec:preliminaries}, we briefly review the definition and quantum algorithmic primitives.
Sec.~\ref{sec:dequantizing-QEVT} provides the formal definition of geometrically local matrices and dequantizes the quantum algorithm for geometrically local matrices.
Sec.~\ref{sec:applications} is devoted to the direct application of our dequantized algorithm for simulating classical dynamics of coupled harmonic oscillators.
In Sec.~\ref{sec:complexity-short-time-GLI} and Sec.~\ref{sec:complexity-long-time-GLI}, we analyze the computational complexity of simulating short-time and long-time dynamics under a geometrically local matrix.
Then we discuss how much quantum advantage is obtained from the simulation problems.

\section{Preliminaries}\label{sec:preliminaries}
Throughout this paper, we use the notation of $N=2^n, n\in\mathbb{N}$, where $N$ denotes the system size and $n$ denotes the number of qubits to represent the system on quantum computers.
Given two probability distributions $p,q: \{1,\dots,N\}\rightarrow[0,1]$, the total variation distance between them is defined as
\begin{equation}
    |p-q|_{\mathrm{tv}} = \frac{1}{2} \sum_{i=1}^{N}|p_i - q_i|.
\end{equation}
For any non-zero vector $u\in \mathbb{C}^N$, we define the probability distribution $p^u: \{1,\dots,N\}\rightarrow[0,1]$ as
\begin{equation}
    p^u_i = \frac{|u_i|^2}{\|u\|^2},
\end{equation}
for each $i\in \{1,\dots,N\}$.

\subsection{Query-access and sampling-access to vectors and matrices}\label{subsec:query}
We introduce two ways of accessing vectors and matrices mainly based on Refs.~\cite{gharibian2022dequantizing, gall2023robust}.
First, we define query-access to a vector.
\begin{dfn}[Query-access to a vector]\label{dfn:query-vec}
We say that we have query-access to a vector $u\in\mathbb{C}^{N}$, which we write $\mathcal{Q}(u)$, if on input $i\in\{1,\dots,N\}$ we can query the entry $u_i$.
We denote the runtime cost of implementing one such query as $\bm{q}(u)$.
\end{dfn}
We define query-access to sparse square matrices which allows us to efficiently compute the non-zero entries and their positions.
\begin{dfn}[Query-access to a sparse square matrix]\label{dfn:query-matrix}
We say that we have query-access to an $s$-sparse square matrix $A\in\mathbb{C}^{N\times N}$, which we write $\mathcal{Q}(A)$, if we have queries $\mathcal{Q}^{\mathrm{row}}(A)$ and $\mathcal{Q}^{\mathrm{col}}(A)$ such that:
\begin{itemize}
    \item on input $(i,l)\in \{1,\dots,N\}\times\{1,\dots,s\}$, the query $\mathcal{Q}^{\mathrm{row}}(A)$ outputs the $l$-th non-zero entry of the $i$-th row of A and its column index if this row has at least $l$ non-zero entries, and outputs an error message otherwise;
    \item on input $(j,l)\in \{1,\dots,N\}\times\{1,\dots,s\}$, the query $\mathcal{Q}^{\mathrm{col}}(A)$ outputs the $l$-th non-zero entry of the $j$-th column of A and its row index if this column has at least $l$ non-zero entries, and outputs an error message otherwise.
\end{itemize}
We denote the runtime cost of implementing one such query as $\bm{q}(A)$.
\end{dfn}
Typically, the cost of implementing one query is polylogarithmic-time with respect to the dimension of vectors and matrices when each entry is stored in a (classical) random-access-memory.

Next, we define sampling-access to a vector.
\begin{dfn}[Sampling-access to a vector]\label{dfn:sampling-vec}
We say that we have $\varepsilon$-sampling-access to a vector $u\in\mathbb{C}^{N}$, which we write $\mathcal{S}_{\varepsilon}(u)$, if we can sample from a distribution $\tilde{p}:\{1,\dots,N\}\rightarrow[0,1]$ such that $|p^u-\tilde{p}|_{\mathrm{tv}} \leq \varepsilon$.
We denote the runtime cost of generating one sample as $\bm{s}(u)$.
\end{dfn}
Combined with query-access, we define sampling-and-query-access, which we mainly use.
\begin{dfn}[Sampling-and-query-access to a vector]\label{dfn:sampling-query-vec}
We say that we have $\varepsilon$-sampling-and-query-access to a vector $u\in\mathbb{C}^{N}$, which we write $\mathcal{SQ}_{\varepsilon}(u)$, if we have $\mathcal{Q}(u)$ and $\mathcal{S}_{\varepsilon}(u)$ and additionally can get the value $\|u\|$.
We denote the maximum among $\bm{q}(u)$, $\bm{s}(u)$ and the cost of obtaining $\|u\|$ as $\bm{sq}(u)$.
When $\varepsilon = 0$, we simply write $\mathcal{SQ}(u)$
\end{dfn}

\subsection{Eigenvalue transformation and singular value transformation}\label{subsec:function}
The eigenvalue transformation (EVT) of a square matrix $A$ with a polynomial $P\in\mathbb{C}[x]$ is defined as a polynomial transformation of $A$.
That is, given a square matrix $A\in\mathbb{C}^{N\times N}$ (possibly non-normal) and a polynomial $P\in\mathbb{C}[x]$ of degree $d$, if we write 
\begin{equation}
P(x) = a_0 + a_1 x + a_2 x^2 + \dots + a_d x^d,
\end{equation}
then EVT of $A$ is given by
\begin{equation}
P(A) = SP(J)S^{-1} = a_0 I + a_1 A + a_2 A^2 + \dots + a_d A^d,
\end{equation}
where $A=SJ S^{-1}$ is the Jordan form decomposition, $S$ is an invertible matrix, and $J$ is the so-called Jordan normal form that is block diagonal (see Ref.~\cite{low2024quantum} for more details about the EVT from the aspect of Jordan form decomposition).
Quantum eigenvalue transformation (QEVT) is a quantum algorithm for implementing $P(A)$ on quantum computers~\cite{low2024quantum, takahira2021quantum, takahira2021quantumalgorithms, an2024laplace}, which has various applications, including linear differential equations~\cite{berry2014high, berry2017quantum, fang2023time, an2023linear, an2023quantum, morales2024quantum} and non-Hermitian physics~\cite{ashida2020non, bender2007faster}.
In this work, we dequantize the QEVT for geometrically local matrices, whose definition is given in Sec.~\ref{subsec:GLI}, since we mainly consider the application to solving linear differential equations.

On the other hand, the singular value transformation (SVT) transforms the singular values of a given matrix rather than its eigenvalues.
Due to some technical issues, SVT is defined in different forms for the cases of odd polynomials and even polynomials.
For any matrix $A\in\mathbb{C}^{M\times N}$ (possibly non-square), singular value decomposition (SVD) is defined as 
\begin{equation}
A=\sum_{i=1}^{\min{(M,N)}} \sigma_i u_i v_i^\dag,
\end{equation}
where $\sigma_i$'s are non-negative real numbers (so-called singular values), $\{u_i\}_i$ are orthonormal vectors in $\mathbb{C}^M$ and $\{v_i\}_i$ are orthonormal vectors in $\mathbb{C}^N$. 
Then SVT of $A$ with an even polynomial $P_\mathrm{even}$ of degree $2d$, which we write
\begin{equation}
P_\mathrm{even}(x) = a_0 + a_2 x^2 + a_4 x^4 + \dots + a_{2d} x^{2d},
\end{equation}
is defined as
\begin{equation}
P_\mathrm{even}(\sqrt{A^\dag A}) = \sum_{i=1}^{\min{(M,N)}} P_\mathrm{even}(\sigma_i) v_i v_i^\dag = a_0 I + a_2 A^\dag A + a_4 (A^\dag A)^2 + \dots + a_{2d} (A^\dag A)^d.
\end{equation}
Similarly, SVT of $A$ with an odd polynomial $P_\mathrm{odd}$ of degree $2d+1$, which we write
\begin{equation}
P_\mathrm{odd}(x) = a_1 x + a_3 x^3 + a_5 x^5 + \dots + a_{2d+1} x^{2d+1} \eqqcolon x P_\mathrm{odd}'(x),
\end{equation}
is defined as
\begin{equation}
A P_\mathrm{odd}'(\sqrt{A^\dag A}) = \sum_{i=1}^{\min{(M,N)}} P_\mathrm{odd}(\sigma_i) u_i v_i^\dag = a_1 A + a_3 A A^\dag A + a_5 A (A^\dag A)^2 + \dots + a_{2d+1} A (A^\dag A)^d.
\end{equation}
Quantum singular value transformation (QSVT) allows efficient polynomial transformations of the singular values of a block-encoded matrix on quantum computers~\cite{gilyen2019quantum, chakraborty2018power}. 
QSVT is significantly more efficient than QEVT and provides a unified framework for describing many existing quantum algorithms~\cite{martyn2021grand}.
However, when solving differential equations, we often want to transform the eigenvalues of a given matrix rather than its singular values, which is a case where QSVT is not applicable. 
In this sense, QEVT has an affinity with applications to differential equation solvers, which makes us focus on QEVT in this work. 
Nonetheless, note that the dequantization techniques we have developed can be naturally extended to QSVT as well.

\section{The dequantized algorithm for the classical linear systems with geometrically local interactions}\label{sec:dequantizing-QEVT}
In this work, we consider classical linear systems determined by first-order linear differential equations of the form:
\begin{equation}
\frac{\mathrm{d}}{\mathrm{d}t} x(t) = A x(t),
\end{equation}
where $x(t)$ is an $N$-element vector, $A$ is an $N\times N$ matrix and $N=2^n$.
If we consider the linear differential equation with higher-order derivatives, we can convert it into a first-order linear differential equation of higher dimensions by introducing auxiliary variables.
We focus on quantum algorithms that simulate such systems using $O(n)$ qubits. 
The goal of quantum algorithms is to prepare a quantum state proportional to the solution $x(t) = e^{At} x(0)$ for a given evolution time $t$.
For instance, quantum computers approximate the time evolution operator by polynomial transformation $P(A)\approx e^{At}$ of the given matrix $A$ using QEVT algorithm~\cite{low2024quantum}.
Thus we consider a general problem of simulating a polynomial transformation of a given matrix $A$.
When $A$ is anti-Hermitian and $\|A\|=\mathrm{poly}(n)$, we can simulate the dynamics efficiently since the degree of polynomial approximation of $e^{At}$ is given by $d=O(\|A\|t)$~\cite{gilyen2019quantum}.
This includes the case of the system of classical coupled oscillators~\cite{babbush2023exponential} and the wave equation~\cite{costa2019quantum},
where the classical systems are mapped into Schr\"{o}dinger equation.
In the case of a general matrix $A$, quantum algorithms do not necessarily simulate the dynamics efficiently since the worst-case runtime is $\Omega(e^t)$~\cite{an2022theory}.
However, under the conditions that $A$ has eigenvalues with non-positive real parts, $\|A\|=\mathrm{poly}(n)$, and $\|e^{At}x(0)\|\geq 1/\mathrm{poly}(n)$, the quantum algorithm works efficiently~\cite{krovi2023improved, berry2024quantum, an2023quantum, low2024quantumlinear}.
Throughout this work, we impose these assumptions on $A$, which restricts our focus to classical systems that can be efficiently simulated by $O(n)$-qubit quantum algorithms whose runtime is at most $\mathrm{poly}(t,n)$.
We also assume that we have a sampling-and-query-access to the initial state $x(0)$.
This assumption ensures that the description of the initial state does not contribute to a quantum advantage.
We then consider two computational problems which are natural for quantum computation: estimating a quantity given by an inner product (such as $\braket{c|x(t)}$) in Sec.~\ref{sec:dequantizing-QEVT} and Sec.~\ref{sec:applications}, and sampling from the quantum state corresponding to $x(t)$ in Sec.~\ref{sec:complexity-short-time-GLI} and Sec.~\ref{sec:complexity-long-time-GLI}.

In this section, we show that, if $A$ has a certain feature of locality, we can simulate the polynomial transformation of $A$ classically using $\mathrm{poly}(n,d)$ resources, where $d$ is the degree of the polynomial.
Here, it is important to emphasize that the locality of $A$ is defined not for qubits, but for indices in the computational basis, which arises naturally when considering differential equations for classical systems.
The complexity of $\mathrm{poly}(n,d)$ is the same as that of quantum computers up to polynomial overhead as long as $d=\mathrm{poly}(n)$.
This leads to our first result that there is no exponential quantum advantage in simulating the short-time ($t=\mathrm{polylog}(N) = \mathrm{poly}(n)$) classical linear dynamics with geometrically local interactions.

\subsection{Geometrically local matrices}\label{subsec:GLI}
Here we define the class of matrices, so-called geometrically local matrix, which is the central concept of our work.
The matrices in this class have a certain kind of locality, which leads to the essential feature of this matrix corresponding to the light cone as shown in Sec.~\ref{subsec:dequantizing-QEVT}.
\begin{dfn}[Geometrically local matrix]\label{dfn:GLI}
Consider the finite-volume lattice with sites labeled as $i=\{1,\dots,N\}$, where the structure of the lattice is determined by some set of bonds. 
Let $d(i,j)$ be the distance between the sites $i$ and $j$, and
$\mathcal{N}(\cdot)\in \mathbb{R}_{+} \rightarrow \mathbb{N}$ be a locality function defined as:
\begin{equation}
\mathcal{N}(r) \coloneqq \max_{i=\{1,\dots,N\}} \left| \left\{ j | d(i,j) \leq r \right\} \right|,
\end{equation}
where $\left| \left\{ j | d(i,j) \leq r \right\} \right|$ is the number of elements contained in the set $\left\{ j | d(i,j) \leq r \right\}$.
Then we say that $A\in\mathbb{C}^{N\times N}$ is an $(r_0,\mathcal{N}(r_0))$-geometrically local matrix if there exists a value $r_0\in\mathbb{R}_{+}$ such that
\begin{equation}
A_{ij} = 0 \quad \text{for} \quad d(i,j) > r_0,
\end{equation}
and $\mathcal{N}(r) = \mathrm{poly}(r)$ for any $r\in \mathbb{R}_{+}$.
When $A$ is Hermitian, we particularly say that it is an $(r_0,\mathcal{N}(r_0))$-geometrically local Hamiltonian.
\end{dfn}
We note that the lattice and its distance function used in the definition of a geometrically local matrix are assumed to be given a priori.
This lattice structure does not necessarily have to represent some physical system since it is introduced only to define distances between sites.
Moreover, the constant $r_0$ does not have to represent the real diameter of locality, just counts the number of sites to which information propagates from a given site when the geometrically local matrix is applied once.
Thus $\mathcal{N}(r_0)$ directly corresponds to the sparsity of the geometrically local matrix.
For example, if we consider $D$-dimensional system, $\mathcal{N}(r_0)=r_0^D$.

Here we provide some specific examples of gemetrically local matrices.
Clearly, the identity matrix $I$ is a $(0,\mathcal{N}(0)=1)$-geometrically local matrix.
One non-trivial example is a tridiagonal matrix.
We say a matrix $A$ as a tridiagonal matrix when $A$ has non-zero elements only on the diagonal, lower diagonal and upper diagonal.
We can consider the one-dimensional system behind this matrix as in Fig.~\ref{fig:tridiagonal_system}, and then we conclude that a tridiagonal matrix is a $(1,\mathcal{N}(1)=3)$-geometrically local matrix.
Similarly, for a general band matrix with bandwidth $b$ such that $A_{i,j}=0$ for $|i-j|>b$, we can consider the one-dimensional system where the site $i$ is connected to at most $2b+1$ sites $i-b, \dots, i-1,i,i+1,\dots,i+b$.
Then such matrix is a $(b,\mathcal{N}(b)=2b+1)$-geometrically local matrix.
In addition, Laplacians for wave equations after the discretization and other matrices corresponding to partial differential equations can be viewed as a special case of geometrically local matrices.
See Sec.~\ref{sec:applications} and Appendix~\ref{appsec:other-applications} for more details.

\begin{figure}
\centerline{
\includegraphics[width=120mm, page=1]{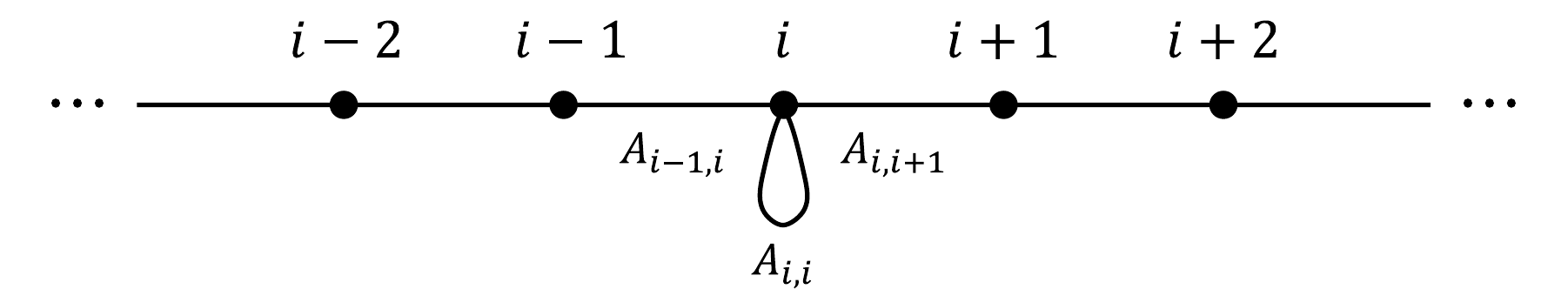}
}
\caption{The one-dimensional system behind a tridiagonal matrix $A$. Each vertex corresponds to some index of $A$ and each edge represents the matrix element $A_{i,j}$. Clearly, we can consider non-symmetric matrices by introducing one-dimensional directed graph.}
\label{fig:tridiagonal_system}
\end{figure}

It is also important to provide an example of \textit{non}-geometrically local matrices.
Consider a balanced binary tree of depth $n$.
In this case, we take $d(i,j)$ to be the distance on the tree, i.e., the number of edges in the shortest path between the vertices $i$ and $j$.
Then the locality function of this graph scales as $\mathcal{N}(r)=\Omega(2^r)$.
Thus the adjacency matrix $A\in\mathbb{R}^{(2^n-1) \times (2^n-1)}$ of a balanced binary tree of depth $n$ is a sparse matrix but is not geometrically local matrix.
This implies that a sparse matrix is not necessarily geometrically local while a geometrically local matrix is always sparse.

\subsection{Dequantizing the quantum eigenvalue transformation for the geometrically local matrix}\label{subsec:dequantizing-QEVT}
Now we provide the dequantized algorithm of the QEVT for geometrically local matrices.
The following lemma is essential for this purpose, which restricts the space that needs to be considered.
\begin{lem}\label{lem:GLI-power}
Suppose that $A\in\mathbb{C}^{N\times N}$ is an $(r_0,\mathcal{N}(r_0))$-geometrically local matrix for some value $r_0\in\mathbb{R}_{+}$.
Then, for any $k\in \mathbb{N}$, the following holds:
\begin{equation}
(A^k)_{ij} = 0 \quad \text{for} \quad d(i,j) > k r_0.
\end{equation}
That is, $A^k$ is a $(k r_0,\mathcal{N}(k r_0))$-geometrically local matrix.
\end{lem}
\begin{proof}
For the $k=1$ case, $A$ is an $(r_0,\mathcal{N}(r_0))$-geometrically local matrix by definition.
Now we prove the contrapositive, i.e.,
\begin{equation}
d(i,j) \leq k r_0 \quad \text{if} \quad (A^{k})_{ij} \neq 0,
\end{equation}
by induction.
Suppose that we have proved for $k-1$ that 
\begin{equation}
d(i,j) \leq (k-1)r_0 \quad \text{if} \quad (A^{k-1})_{ij} \neq 0.
\end{equation}
Assuming that $i$ and $l$ are indices such that $(A^k)_{il}\neq 0$, 
there exists an index $j$ such that $(A^{k-1})_{ij}\neq 0$ and $A_{jl}\neq 0$ since we have $(A^k)_{il} = \sum_j (A^{k-1})_{ij} A_{jl}$.
By the induction hypothesis, we get $d(i,j)\leq (k-1)r_0$ and $d(j,l)\leq r_0$. 
From the metric axioms, we obtain
\begin{equation}
d(i,l) \leq d(i,j) + d(j,l) \leq k r_0.
\end{equation}
This implies that $(A^k)_{ij} = 0$ if $d(i,j)>k r_0$.
\end{proof}

Using Lemma~\ref{lem:GLI-power}, we can construct a query-access to a vector generated by acting a polynomial transformation of a geometrically local matrix on a given vector.
\begin{lem}\label{lem:query-GLI-poly}
Let $P\in \mathbb{C}[x]$ be a polynomial of degree $d$.
Then there exists a $O\left(d^2 \mathcal{N}(d r_0)\mathcal{N}(r_0) \bm{q}(A) + d \mathcal{N}(d r_0) \bm{q}(u) \right)$-time classical algorithm that given
\begin{itemize}
    \item query-access to an $(r_0,\mathcal{N}(r_0))$-geometrically local matrix $A\in\mathbb{C}^{N\times N}$,
    \item query-access to a vector $u\in \mathbb{C}^{N}$,
    \item an index $i\in \{1,\dots,N\}$,
\end{itemize}
outputs the $i$-th entry of $P(A)u$.
\end{lem}
\begin{proof}
Supposing that the polynomial $P$ can be written as
\begin{equation}
P(x) = a_0 + a_1 x + a_2 x^2 + \dots + a_k x^d,
\end{equation}
we have
\begin{equation}
P(A)u = a_0 u + a_1 A u + a_2 A^2 u + \dots + a_k A^d u.
\end{equation}
Now we consider how to compute the $i$-th entry of each term $e_i^\dag A^k u$ for $k\in\{1,\dots,d\}$, where $e_i$ is a unit vector with a 1 in the $i$-th position and zeros elsewhere.
First, we query all entries of the $i$-th row of $A$ and get at most $\mathcal{N}(r_0)$ non-zero entries and their positions using the query $\mathcal{Q}^{\mathrm{row}}(A)$ from Definition~\ref{dfn:query-matrix}. 
Then we can compute all non-zero entries of $e_i^\dag A$ at the cost of $\bm{q}(A)\cdot \mathcal{N}(r_0)$-time.
Note that the number of non-zero entries of $e_i^\dag A$ is also at most $\mathcal{N}(r_0)$.
Next, we assume that we have already computed all non-zero entries of $e_i^\dag A^{k-1}$.
From Lemma~\ref{lem:GLI-power}, the number of non-zero entries of $e_i^\dag A^{k-1}$ is at most $\mathcal{N}((k-1)r_0)$.
Here let $j_1,\dots , j_l$ denote the positions of non-zero entries of $e_i^\dag A^{k-1}$ with $l\leq \mathcal{N}((k-1)r_0)$.
In order to compute all non-zero entries of $e_i^\dag A^k$, we only have to query all entries of the $j_1,\dots , j_l$-th rows of $A$ and get at most $\mathcal{N}(k r_0)$ non-zero entries and their positions.
This procedure can be accomplished at the cost of $\mathcal{N}((k-1) r_0)\cdot \mathcal{N}(r_0)\cdot\bm{q}(A)$-time.
Putting it altogether, the runtime for calculating all non-zero entries of $e_i^\dag A^r$ is upper-bounded by
\begin{align}
&\mathcal{N}(r_0) + \mathcal{N}(r_0)\cdot \mathcal{N}(r_0) \cdot \bm{q}(A) + \mathcal{N}(2 r_0)\cdot \mathcal{N}(r_0) \cdot \bm{q}(A) + \dots + \mathcal{N}((k-1) r_0)\cdot \mathcal{N}(r_0)\cdot\bm{q}(A) \\
=& \sum_{j=0}^{k-1}\mathcal{N}(j r_0)\cdot \mathcal{N}(r_0)\cdot\bm{q}(A).
\end{align}
Finally, we calculate $e_i^\dag A^k u$ by computing $\mathcal{N}(k r_0)$ elements of $u$ using algorithm $\mathcal{Q}(u)$ from Definition~\ref{dfn:query-vec} at the cost of $\mathcal{N}(k r_0)\cdot \bm{q}(u)$-time and taking the inner product at the cost of
$\mathcal{N} (k r_0)$-time.

The resulting runtime for calculating $e_i^\dag P(A) u$ is upper-bounded by
\begin{align}
&\bm{q}(u) + \sum_{k=1}^{d}\left( \sum_{j=0}^{k-1}\mathcal{N}(j r_0)\cdot \mathcal{N}(r_0)\cdot\bm{q}(A) + \mathcal{N}(k r_0)\cdot \bm{q}(u) + \mathcal{N} (k r_0) \right) \\
\leq & \bm{q}(u) + \sum_{k=1}^{d}\left( k\mathcal{N}(k r_0)\cdot \mathcal{N}(r_0)\cdot\bm{q}(A) + \mathcal{N}(k r_0)\cdot (\bm{q}(u)+1) \right) \\
\leq & \bm{q}(u) + d\left( d\mathcal{N}(d r_0)\cdot \mathcal{N}(r_0)\cdot \bm{q}(A) + \mathcal{N}(d r_0)\cdot (\bm{q}(u)+1) \right) \\
\leq & O\left(d^2 \mathcal{N}(d r_0)\mathcal{N}(r_0) \bm{q}(A) + d \mathcal{N}(d r_0) \bm{q}(u) \right).
\end{align}
\end{proof}

The central idea of dequantization for sparse matrices is that estimation of the inner product of two vectors $u$ and $v$ can be accomplished by sampling from $p^v$ combined with a query-access to $u$.
Roughly speaking, sampling an index $i\in\{i,\dots,N\}$ according to the probability distribution $p^v$ and computing $u_i\|v\|^2/v_i$ leads to the $v^\dag u$ as an expectation value.
Ref.~\cite{gall2023robust} provided the robust version of the inner product estimation as follows.
\begin{lem}[\cite{gall2023robust}]\label{lem:sample-query-est}
Suppose that we are given
\begin{itemize}
    \item query-access to a vector $w\in\mathbb{C}^N$ such that $\|w\|\leq 1$,
    \item $\zeta$-sampling-and-query-access to a vector $v\in\mathbb{C}^N$ such that $\|v\|\leq 1$.
\end{itemize}
Then we can estimate $\hat{z}\in\mathbb{C}$ such that
\begin{equation}
\left|\hat{z} - v^\dag w \right| \leq \varepsilon,
\end{equation}
classically with a probability at least $1-\delta$ in $O(\log{(1/\delta)}(\bm{q}(w)+\bm{sq}(v))/\varepsilon^2)$-time for $\zeta\leq \varepsilon/9$.
\end{lem}

Here we formulate the problem of the eigenvalue transformation for the geometrically local matrix.
\begin{problem}[$\mathsf{EVT}\textsf{-}\mathsf{GL}(r_0,\mathcal{N}(r_0),\varepsilon,\zeta)$]
Suppose that we are given
\begin{itemize}
    \item query-access to a vector $u\in\mathbb{C}^N$ such that $\|u\|\leq 1$,
    \item query-access to an $(r_0,\mathcal{N}(r_0))$-geometrically local matrix $A\in \mathbb{C}^{N\times N}$ for some value $r_0\in\mathbb{R}_{+}$,
    \item $\zeta$-sampling-and-query-access to a vector $v\in\mathbb{C}^N$ such that $\|v\|\leq 1$.
\end{itemize}
Given a polynomial $P\in \mathbb{C}[x]$ of degree $d$ with $|P(x)|\leq 1$ for any $x\in [-\|A\|,\|A\|]$, the goal is to output an estimate $\hat{z}\in\mathbb{C}$ such that
\begin{equation}\label{eq:estimate-of-evt-GLI}
\left|\hat{z} - v^\dag P(A) u \right| \leq \varepsilon.
\end{equation}
\end{problem}

Here is the statement of the main result:
\begin{thm}\label{thm:dequantize-qevt}
For any $r_0\in\mathbb{R}_{+}$, any $\varepsilon \in (0,1]$ and a function $\mathcal{N}(\cdot)\in \mathbb{R}_{+} \rightarrow \mathbb{N}$ defined in Definition~\ref{dfn:GLI}, the problem $\mathsf{EVT}\textsf{-}\mathsf{GL}(r_0,\mathcal{N}(r_0),\varepsilon,\zeta)$ can be solved classically with a probability at least $1-\delta$ in $O(\log{(1/\delta)}(d^2 \mathcal{N}(d r_0)\mathcal{N}(r_0) \bm{q}(A) + d \mathcal{N}(d r_0) \bm{q}(u) +\bm{sq}(v))/\varepsilon^2)$-time for $\zeta\leq \varepsilon/9$.
\end{thm}
\begin{proof}
From Lemma~\ref{lem:query-GLI-poly}, we can construct a query-access to $P(A) u$ at a cost of 
\begin{equation}\label{eq:runtime-query-PAu}
\bm{q} (P(A) u) = O\left(d^2 \mathcal{N}(d r_0)\mathcal{N}(r_0) \bm{q}(A) + d \mathcal{N}(d r_0) \bm{q}(u) \right).
\end{equation}
Lemma~\ref{lem:sample-query-est} implies that, for $\zeta\leq \varepsilon/9$, we can calculate an estimate such that Eq.~\eqref{eq:estimate-of-evt-GLI} classically with a probability at least $1-\delta$ at a cost of 
\begin{equation}
O\left( \frac{\log{(1/\delta)}}{\varepsilon^2}\left(\bm{q}(P(A) u)+\bm{sq}(v)\right) \right).
\end{equation}
Combined with Eq.~\eqref{eq:runtime-query-PAu}, the final runtime of this procedure is
\begin{equation}
O\left( \frac{\log{(1/\delta)}}{\varepsilon^2}\left( d^2 \mathcal{N}(d r_0)\mathcal{N}(r_0) \bm{q}(A) + d \mathcal{N}(d r_0) \bm{q}(u) +\bm{sq}(v)\right) \right).
\end{equation}
\end{proof}

According to Theorem~\ref{thm:dequantize-qevt} and Definition~\ref{dfn:GLI}, our dequantized algorithm of QEVT has the complexity of $\mathrm{poly}(d)$ for the degree $d$ of the polynomial while the previous work~\cite{gharibian2022dequantizing} that focused on general sparse matrices has the complexity of $\mathrm{exp}(d)$.
The key point is that, for a geometrically local matrix, $A^d$ has a sparsity bounded by $\mathcal{N}(dr_0)=\mathrm{poly}(d)$, whereas, for a general $s$-sparse matrix, the sparsity of $A^d$ can grow as $s^d$ in the worst case.
This leads to a significant reduction of the cost of the classical algorithm and broadens the parameter region that is tractable for classical computers.
We note that, if $A^d$ is assumed to have polynomial sparsity in $d$, the dequantized algorithm for general sparse matrices~\cite{gharibian2022dequantizing} can also have polynomial complexity in $d$.
Thus our contribution is to identify geometrically local matrices as a natural class for which this property is guaranteed.
Also note that our result is different from dequantization results based on low-rank assumptions since geometrically local matrices are not necessarily low-rank.

\section{Applications of the dequantized algorithm for the geometrically local matrix}\label{sec:applications}
In this section, we provide the classical algorithms that simulate the quantum algorithms for solving the specific problems related to geometrically local matrices using the results in Sec.~\ref{sec:dequantizing-QEVT}.
Here we provide only the application for simulating a classical system of coupled harmonic oscillators while we can also apply our dequantized algorithm for simulating wave equations~\cite{costa2019quantum, schade2024quantum, bosch2024quantum}, advection equations~\cite{sato2024hamiltonian} and single particle Schr\"{o}dinger equations~\cite{wiesner1996simulations, zalka1998efficient, benenti2008quantum, childs2022quantum} which we postpone in Appendices~\ref{appsubsec:wave},~\ref{appsubsec:advection} and~\ref{appsubsec:schrodinger}.

\subsection{Application to the simulation of classical dynamics with geometrically local interaction}\label{subsec:classical}
As in Ref.~\cite{babbush2023exponential}, we consider a classical system of coupled harmonic oscillators with $N=2^n$ point mass $m_1,\dots,m_{N}$ which are coupled with each other by springs.
We define $\kappa_{ij}=\kappa_{ji}\geq 0$ as the spring constants which couple $i$-th and $j$-th oscillators and $\kappa_{ii}$ as the spring constant which connects $i$-th oscillator to a wall.
In the case of $D=1$-dimension, the displacements and velocities of the point masses are given by $x(t) = (x_1(t),\dots,x_{N}(t))^{T}$ and $\dot{x}(t) = (\dot{x}_1(t),\dots,\dot{x}_{N}(t))^{T}$ respectively.
We can also address the higher dimensional case similarly.
In two-dimensional case, for example, we just use $x_1(t)$ as the first coordinate of the first point mass and $x_2(t)$ as the second coordinate of the first point mass, and handle other point masses in the same manner.
After assigning the indices in this way, we redefine the dimension of vectors $x(t)$ and $\dot{x}(t)$ as $N$

In the harmonic approximation, the dynamics of the coupled oscillators is described by Newton's equation:
\begin{equation}\label{eq:eom}
m_i \ddot{x}_i(t) = \sum_{j\neq i} \kappa_{ij} (x_j(t) - x_i(t)) -\kappa_{ii} x_i(t),
\end{equation}
for all $i\in\{1,\dots,N\}$. 
Now our goal is to simulate the time evolution of $x(t)$ and $\dot{x}(t)$ given initial conditions $x(0)$ and $\dot{x}(0)$.

In order to solve this problem on quantum computers, the authors of Ref.~\cite{babbush2023exponential} transformed Eq.~\eqref{eq:eom} into Schr\"{o}dinger equation.
Writing Eq.~\eqref{eq:eom} in matrix form, we obtain $M\ddot{x}(t) = -F x(t)$, where $M$ is a $N\times N$ diagonal matrix with entries $m_i\geq 0$ and $F$ is a $N\times N$ matrix whose diagonal and off-diagonal entries are $f_{ii}=\sum_j \kappa_{ij}$ and $f_{ij}=-\kappa_{ij}$ respectively.
Note that $M\geq 0$ and $F\geq 0$ are real and positive semidefinite.
Changing the variables as $y(t) = \sqrt{M}x(t)$, we get $\ddot{y}(t) = -A y(t)$, where $A=\sqrt{M}^{-1} F \sqrt{M}^{-1} \geq 0$ is real and positive semidefinite.
We can transform this second-order differential equation into the first-order differential equation as:
\begin{equation}
\frac{\mathrm{d}}{\mathrm{d}t} \left( \dot{y}(t) + i \sqrt{A} y(t) \right) = i \sqrt{A} \left( \dot{y}(t) + i \sqrt{A} y(t) \right) .
\end{equation}
Then its solution is given by
\begin{equation}
\dot{y}(t) + i \sqrt{A} y(t) = e^{i \sqrt{A} t} \left( \dot{y}(0) + i \sqrt{A} y(0) \right) .
\end{equation}
However, we do not have direct access to $\sqrt{A}$ and simulating the time evolution under $\sqrt{A}$ requires additional costs.
To address this issue, in Ref.~\cite{babbush2023exponential} (see also Ref.~\cite{costa2019quantum}), the authors introduced the generalization of $\sqrt{A}$:
\begin{equation}\label{eq:dilated-hamiltonian-mass-spring}
H \coloneqq 
\begin{pmatrix}
   0 & B \\
   B^\dag & 0
\end{pmatrix},
\end{equation}
where $B$ is a $N\times N_\mathrm{col}$ matrix such that $BB^\dag =A$ for some integer $N_\mathrm{col}$.
Note that 
\begin{equation}\label{eq:Hsquare}
H^2 \coloneqq 
\begin{pmatrix}
   B B^\dag & 0 \\
   0 & B^\dag B
\end{pmatrix}
=
\begin{pmatrix}
   A & 0 \\
   0 & B^\dag B
\end{pmatrix},
\end{equation}
which allows us to use $H\in \mathbb{C}^{(N+N_\mathrm{col})\times (N+N_\mathrm{col})}$ as a square root of $A$.
Now it is easy to see that 
\begin{equation}
\ket{\psi(t)} \coloneqq
\frac{1}{\sqrt{2E}}
\begin{pmatrix}
   \dot{y}(t) \\
   i B^\dag y(t)
\end{pmatrix},
\end{equation}
where $\sqrt{2E}$ is a normalization factor, is a solution of the Schr\"{o}dinger equation
\begin{equation}
\frac{\mathrm{d}}{\mathrm{d}t} \ket{\psi(t)} = iH \ket{\psi(t)}.
\end{equation}
Hence, we obtain
\begin{equation}
\frac{1}{\sqrt{2E}}
\begin{pmatrix}
   \dot{y}(t) \\
   i B^\dag y(t)
\end{pmatrix}
 = e^{iHt} 
 \frac{1}{\sqrt{2E}}
 \begin{pmatrix}
   \dot{y}(0) \\
   i B^\dag y(0)
\end{pmatrix}
= \left( \cos{(Ht)} + i \sin{(Ht)} \right)
\frac{1}{\sqrt{2E}}
 \begin{pmatrix}
   \dot{y}(0) \\
   i B^\dag y(0)
\end{pmatrix}.
\end{equation}
In the above discussion, the matrix $B$ may be an arbitrary matrix such that $BB^\dag =A$.
In this work, we choose
\begin{equation}\label{eq:B-dag}
B^\dag \ket{i} = \sum_{j\geq i}\sqrt{\frac{\kappa_{ij}}{m_i}} \ket{i}\ket{j} - \sum_{j< i}\sqrt{\frac{\kappa_{ij}}{m_i}} \ket{j}\ket{i}
\end{equation}
as in Ref.~\cite{babbush2023exponential}, where $N_\mathrm{col}=N(N+1)/2$.
This choice allows us to obtain
\begin{equation}
B^\dag y(t) = B^\dag \sqrt{M} x(t) = \sum_{i}\sqrt{\kappa_{ii}} x_i(t) \ket{i}\ket{i} + \sum_{j> i}\sqrt{\kappa_{ij}} (x_i(t)-x_j(t)) \ket{i}\ket{j},
\end{equation}
which implies that
\begin{equation}
\frac{1}{2} (B^\dag y(t))^\dag B^\dag y(t) = \frac{1}{2} \sum_{i}\kappa_{ii} x_i(t)^2 + \frac{1}{2} \sum_{j> i}\kappa_{ij} (x_i(t)-x_j(t))^2 ,
\end{equation}
is the potential energy.

Now we are ready to dequantize the simulation of the classical system with geometrically local interactions.
In order to introduce the distance between each point mass, we set each point mass on some lattice (e.g., square lattice) with lattice spacing $1$ (: constant).
We have to note that this lattice is only used to define geometrically local interaction, and the distance between $i$-th and $j$-th sites does not imply the `real' distance between $i$-th and $j$-th point masses. 
In other words, this lattice is only used to count the number of masses that is connected to one given mass by a spring.
Letting the distance on lattice between $i$-th and $j$-th point masses be $d(i,j)$, we can define the classical system with geometrically local interactions as the system with springs such that
\begin{equation}\label{eq:locality}
\kappa_{ij} = \kappa_{ji} = 0 \quad \text{for} \quad d(i,j)>r_0,
\end{equation}
where $r_0 \leq \mathrm{poly}(n)$ denotes the spatial locality of the system.
From Definition~\ref{dfn:GLI}, it is easy to see that $F$ and $A$ are $(r_0,\mathcal{N}(r_0))$-geometrically local matrices, where $\mathcal{N}(r_0) \coloneqq \max_{i=\{1,\dots,N\}} \left| \left\{ j | d(i,j) \leq r_0 \right\} \right| \leq \mathrm{poly}(r_0)$.
Remember that $\mathcal{N}(r_0) = O(r_0^D)$ if we solve the $D$-dimensional system.
In Theorem~\ref{thm:classical-time}, we solve the problem related to time evolution of this classical system.

Here we have to note that the dilated Hamiltonian $H$ in Eq.~\eqref{eq:dilated-hamiltonian-mass-spring} is also a geometrically local Hamiltonian.
Considering the lattice as in the above discussion, we set the position $x_i$ and velocity $\dot{x}_i$ of $i$-th mass to the same site $i$.
Furthermore, we define new sites between all the pairs of sites of the lattice and interactions between them.
For example, for $i$-th and $j$-th sites of the lattice, we define the new site $(i,j)$ and define two interactions: one is between $i$ and $(i,j)$ and the other is between $(i,j)$ and $j$.
Then we place the spring $\kappa_{ij}$ on the site $(i,j)$ if $i$-th mass is connected to $j$-th mass.
Using the definition of $B^\dag$ in Eq.~\eqref{eq:B-dag}, the number of masses and springs that are connected to one given mass or spring is given by $O(r_0^{2D})$ for the locality $r_0$.
Thus $H$ is $(r_0,\mathcal{N}(r_0)=O(r_0^{2D}))$-geometrically local matrix.
This leads to the complexity of the dequantized algorithm of $O(t^{2+2D})$ while the complexity of the algorithm in Theorem~\ref{thm:classical-time} is $O(t^{2+D})$.
To reduce the complexity, we apply our dequantization technique to $A$ rather than $H$.

\begin{thm}\label{thm:classical-time}
Let $t\in \mathbb{R}$ be an evolution time, $r_0\geq 0$ be a spatial locality defined in~\eqref{eq:locality} and $\varepsilon\in(0,1]$ be a precision.
Suppose that we are given
\begin{itemize}
\item query-access to $\dot{x}(0)\in \mathbb{R}^N$,
\item query-access to $x(0)\in \mathbb{R}^N$,
\item query-access to $m\in \mathbb{R}^N$, where $m_i$ is the mass of $i$-th point mass for all $i\in\{1,\dots,N\}$,
\item query-access to $K\in \mathbb{R}^{N\times N}$, where $K_{ij}=\kappa_{ij}$,
\item total energy: $E= \frac{1}{2} \sum_{i} m_i \dot{x}_i(0)^2 + \frac{1}{2} \sum_{i}\kappa_{ii} x_i(0)^2 + \frac{1}{2} \sum_{j> i}\kappa_{ij} (x_i(0)-x_j(0))^2$,
\item $\zeta$-sampling-and-query-access to $v\in \mathbb{C}^{N+N(N+1)/2}$.
\end{itemize}
Then, for $\zeta\leq \varepsilon/18$, we can calculate
\begin{equation}
v^\dag \cdot \psi(t)
=
v^\dag \cdot 
\frac{1}{\sqrt{2E}}
\begin{pmatrix}
   \sqrt{M} \dot{x}(t) \\
   i B^\dag \sqrt{M} x(t)
\end{pmatrix},
\end{equation}
classically with a probability at least $1-\delta$ and an additive error $\varepsilon$.
The runtime of the classical algorithm is
\begin{equation}
\begin{split}
O&\left( \frac{\log{(1/\delta)}}{\varepsilon^2}\left( d_{\mathrm{exp}}^2 \mathcal{N}(d_{\mathrm{exp}} r_0)\mathcal{N}(r_0) \left( \bm{q}(K)\mathcal{N}(r_0) + 2\bm{q}(m) \right) \right. \right. \\
& \hspace{4cm} \left. \vphantom{\frac{\log{(1/\delta)}}{\varepsilon^2}} \left. \vphantom{d_{\mathrm{exp}}^2\mathcal{N}(d_{\mathrm{exp}} r_0)}
+ d_{\mathrm{exp}} \mathcal{N}(d_{\mathrm{exp}} r_0) \left( \bm{q}(\dot{x}(0)) + \bm{q}(x(0)) + 2\bm{q}(m) \right) +\bm{sq}(v)\right) \right),
\end{split}
\end{equation}
where $d_{\mathrm{exp}} \coloneqq O( t \sqrt{\mathcal{N}(r_0)\kappa_{\mathrm{max}}/m_{\mathrm{min}}} + \log(1/\varepsilon))$, or emphasizing the dependence on $t$, $\varepsilon$ and $\delta$, given the spatial dimension $D$,
\begin{equation}
\tilde{O}\left( \frac{t^{2+D}\log{(1/\delta)}}{\varepsilon^2} \right).
\end{equation}
\end{thm}
\begin{proof}
For some $\alpha\in \mathbb{R}$, let $P_{\mathrm{exp}}\in\mathbb{C}[x]$ be the polynomial approximation of exponential function $e^{itx}$ on the domain $[-\alpha,\alpha]$ such that
\begin{equation}
\left\| P_{\mathrm{exp}}(x) - e^{itx} \right\|_{[-\alpha,\alpha]} \leq \varepsilon, \quad \left\| P_{\mathrm{exp}}(x) \right\|_{[-\alpha,\alpha]} \leq 1.
\end{equation}
According to Ref.~\cite{gilyen2019quantum}, there exists such an polynomial approximation of degree $O(\alpha t + \log(1/\varepsilon))$.
Using this polynomial with $\alpha \geq \|H\|$, we obtain $\left\| P_{\mathrm{exp}}(H) - e^{iHt} \right\| \leq \varepsilon$ and $\left\| P_{\mathrm{exp}}(H) \right\| \leq 1$.

Here $\|H\|$ can be upper-bounded as $\|H\|\leq \sqrt{\mathcal{N}(r_0)\kappa_{\mathrm{max}}/m_{\mathrm{min}}}$ (see Ref.~\cite{babbush2023exponential}), which enables us to $\varepsilon/2$-approximate $e^{iHt}$ with $P_{\mathrm{exp}}(H)$ of degree $d_{\mathrm{exp}} \coloneqq O( t \sqrt{\mathcal{N}(r_0)\kappa_{\mathrm{max}}/m_{\mathrm{min}}} + \log(1/\varepsilon))$.
Note that if we have an $\varepsilon/2$-approximation of $e^{iHt}$ and estimate $v^\dag P_{\mathrm{exp}}(H) \psi(0)$ with additive error $\varepsilon/2$, then the overall error is upper-bounded by $\varepsilon$.

Now we decompose $P_{\mathrm{exp}}(x)$ into an even polynomial $P_{\mathrm{cos}}(x^2)$ and an odd polynomial $ixP_{\mathrm{sin}}(x^2)$, where $P_{\mathrm{cos}}(x^2)$ and $xP_{\mathrm{sin}}(x^2)$ are real polynomials which approximate $\cos(x)$ and $\sin(x)$ respectively.
Using Eq.~\eqref{eq:Hsquare}, the polynomial transformation of $H$ can be written as
\begin{align}
P_{\mathrm{exp}}(H) &= P_{\mathrm{cos}}(H^2) +i H P_{\mathrm{sin}}(H^2) \\
&= 
\begin{pmatrix}
   P_{\mathrm{cos}}(A) & 0 \\
   0 & P_{\mathrm{cos}}(B^\dag B)
\end{pmatrix}
+i
\begin{pmatrix}
   0 & B \\
   B^\dag & 0
\end{pmatrix}
\begin{pmatrix}
   P_{\mathrm{sin}}(A) & 0 \\
   0 & P_{\mathrm{sin}}(B^\dag B)
\end{pmatrix}\\
&= 
\begin{pmatrix}
   P_{\mathrm{cos}}(A) & i B P_{\mathrm{sin}}(B^\dag B) \\
   i B^\dag P_{\mathrm{sin}}(A) & P_{\mathrm{cos}}(B^\dag B)
\end{pmatrix}.
\end{align}
This leads to
\begin{align}
P_{\mathrm{exp}}(H) \psi(0) 
&= 
\begin{pmatrix}
   P_{\mathrm{cos}}(A) & i B P_{\mathrm{sin}}(B^\dag B) \\
   i B^\dag P_{\mathrm{sin}}(A) & P_{\mathrm{cos}}(B^\dag B)
\end{pmatrix}
\frac{1}{\sqrt{2E}}
\begin{pmatrix}
   \sqrt{M} \dot{x}(0) \\
   i B^\dag \sqrt{M} x(0)
\end{pmatrix}\\
&=
\frac{1}{\sqrt{2E}}
\begin{pmatrix}
P_{\mathrm{cos}}(A) \sqrt{M} \dot{x}(0) -  B P_{\mathrm{sin}}(B^\dag B) B^\dag \sqrt{M} x(0) \\
iB^\dag P_{\mathrm{sin}}(A) \sqrt{M} \dot{x}(0) + i P_{\mathrm{cos}}(B^\dag B) B^\dag \sqrt{M} x(0)
\end{pmatrix}\\
&=
\frac{1}{\sqrt{2E}}
\begin{pmatrix}
P_{\mathrm{cos}}(A) \sqrt{M} \dot{x}(0) - P_{\mathrm{sin},v}(A) \sqrt{M} x(0) \\
iB^\dag (P_{\mathrm{sin}}(A) \sqrt{M} \dot{x}(0) + P_{\mathrm{cos},x}(A) \sqrt{M} x(0))
\end{pmatrix},
\end{align}
where we define $P_{\mathrm{sin},v}(A)\coloneqq B P_{\mathrm{sin}}(B^\dag B) B^\dag$ and $B^\dag P_{\mathrm{cos},x}(A) \coloneqq P_{\mathrm{cos}}(B^\dag B) B^\dag$.
Recalling $\| P_{\mathrm{exp}}(H) \psi(0) \| \leq 1$, we can apply Lemma~\ref{lem:sample-query-est} to a query-access to $P_{\mathrm{exp}}(H) \psi(0)$.
Thus we discuss how to construct a query-access to $P_{\mathrm{cos}}(A) \sqrt{M} \dot{x}(0) - P_{\mathrm{sin},v}(A) \sqrt{M} x(0)$ and $iB^\dag (P_{\mathrm{sin}}(A) \sqrt{M} \dot{x}(0) + P_{\mathrm{cos},x}(A) \sqrt{M} x(0))$ below.

First, we construct a query access to $\sqrt{M} \dot{x}(0)$ and $\sqrt{M} x(0)$.
The $i$-th entries of $\sqrt{M} \dot{x}(0)$ and $\sqrt{M} x(0)$ can be computed by querying the $i$-th entries of $\dot{x}(0)$, $x(0)$ and $m$. 
Thus we can compute each entry of $\sqrt{M} \dot{x}(0)$ and $\sqrt{M} x(0)$ classically in 
\begin{equation}
\bm{q}(\sqrt{M} \dot{x}(0)) = \bm{q}(m) + \bm{q}(\dot{x}(0)), \quad \bm{q}(\sqrt{M} x(0)) = \bm{q}(m) + \bm{q}(x(0))
\end{equation}
-time respectively.
Next, we consider a query-access to $A=\sqrt{M}^{-1} F \sqrt{M}^{-1}$.
A query-access to $F$ can be constructed from a query-access to $K$ since the position of non-zero entries of $F$ is the same as that of $K$ and $f_{ii}=\sum_j \kappa_{ij}$ can be computed using $\mathcal{Q}^{\mathrm{row}}(K)$, $\mathcal{N}(r_0)$ times, which leads to $\bm{q}(F) = \bm{q}(K)\mathcal{N}(r_0)$.
Also the position of non-zero entries of $A$ is the same as that of $F$ and the $(i,j)$-th entry of $A$ can be written as $a_{ij}=f_{ij}/\sqrt{m_i m_j}$.
Combining query-accesses to $F$ and $m$, we can compute the $(i,j)$-th entry of $A$ classically in 
\begin{equation}
\bm{q}(A) = \bm{q}(F) + 2\bm{q}(m) = \bm{q}(K)\mathcal{N}(r_0) + 2\bm{q}(m)
\end{equation}
-time.
Now let us give a query-access to $P_{\mathrm{cos}}(A) \sqrt{M} \dot{x}(0) - P_{\mathrm{sin},v}(A) \sqrt{M} x(0)$.
Considering that the degree of $P_{\mathrm{cos}}$ and $P_{\mathrm{sin},v}$ is $d_{\mathrm{exp}} = O( t \sqrt{\mathcal{N}(r_0)\kappa_{\mathrm{max}}/m_{\mathrm{min}}} + \log(1/\varepsilon))$, from Lemma~\ref{lem:query-GLI-poly}, we obtain
\begin{align}
& \bm{q}\left( P_{\mathrm{cos}}(A) \sqrt{M} \dot{x}(0) - P_{\mathrm{sin},v}(A) \sqrt{M} x(0) \right) \\
=& \bm{q} \left( P_{\mathrm{cos}}(A) \sqrt{M} \dot{x}(0) \right) + \bm{q} \left( P_{\mathrm{sin},v}(A) \sqrt{M} x(0) \right) \\
\begin{split}
    =&O\left(d_{\mathrm{exp}}^2 \mathcal{N}(d_{\mathrm{exp}} r_0)\mathcal{N}(r_0) \bm{q}(A) + d_\mathrm{exp} \mathcal{N}(d_\mathrm{exp} r_0) \bm{q}(\sqrt{M} \dot{x}(0)) \right) \\
    &\hspace{4cm}+ O\left(d_{\mathrm{exp}}^2 \mathcal{N}(d_{\mathrm{exp}} r_0)\mathcal{N}(r_0) \bm{q}(A) + d_{\mathrm{exp}} \mathcal{N}(d_{\mathrm{exp}} r_0) \bm{q}(\sqrt{M} x(0)) \right)
\end{split}\\
=& O\left(d_{\mathrm{exp}}^2 \mathcal{N}(d_{\mathrm{exp}} r_0)\mathcal{N}(r_0) \left( \bm{q}(K)\mathcal{N}(r_0) + 2\bm{q}(m) \right) + d_{\mathrm{exp}} \mathcal{N}(d_{\mathrm{exp}} r_0) \left( \bm{q}(\dot{x}(0)) + \bm{q}(x(0)) + 2\bm{q}(m) \right) \right). \label{eq:query-cost-cos-sinv}
\end{align}
Finally, we provide a query-access to $iB^\dag (P_{\mathrm{sin}}(A) \sqrt{M} \dot{x}(0) + P_{\mathrm{cos},x}(A) \sqrt{M} x(0))$.
$B^\dag$ defined in Eq.~\eqref{eq:B-dag} has at most two entries in each row.
This can be seen from the fact that, considering an $(i,j)$-th row of $B^\dag$, the two non-zero entries $\sqrt{\kappa_{ij}/m_i}$ and $\sqrt{\kappa_{ij}/m_j}$ lie in $i$-th and $j$-th columns respectively if $j>i$, or the one non-zero entry $\sqrt{\kappa_{ii}/m_i}$ lies in $i$-th column if $i=j$.
Thus, given an index $(i,j)$ for $j\leq i$, we can calculate the $(i,j)$-th entry of $iB^\dag (P_{\mathrm{sin}}(A) \sqrt{M} \dot{x}(0) + P_{\mathrm{cos},x}(A) \sqrt{M} x(0))$ using at most two queries to $K$, $m$ and $P_{\mathrm{sin}}(A) \sqrt{M} \dot{x}(0) + P_{\mathrm{cos},x}(A) \sqrt{M} x(0)$.
Noting that the the classical runtime for querying $P_{\mathrm{sin}}(A) \sqrt{M} \dot{x}(0) + P_{\mathrm{cos},x}(A) \sqrt{M} x(0)$ is the same as Eq.~\eqref{eq:query-cost-cos-sinv}, we obtain
\begin{align}
& \bm{q}\left( iB^\dag (P_{\mathrm{sin}}(A) \sqrt{M} \dot{x}(0) + P_{\mathrm{cos},x}(A) \sqrt{M} x(0)) \right) \\
=& 2\bm{q}(K) + 2\bm{q}(m) + 2\bm{q}\left( P_{\mathrm{sin}}(A) \sqrt{M} \dot{x}(0) + P_{\mathrm{cos},x}(A) \sqrt{M} x(0) \right) \\
=& O\left(d_{\mathrm{exp}}^2 \mathcal{N}(d_{\mathrm{exp}} r_0)\mathcal{N}(r_0) \left( \bm{q}(K)\mathcal{N}(r_0) + 2\bm{q}(m) \right) + d_{\mathrm{exp}} \mathcal{N}(d_{\mathrm{exp}} r_0) \left( \bm{q}(\dot{x}(0)) + \bm{q}(x(0)) + 2\bm{q}(m) \right) \right).
\end{align}
Putting it all together, 
\begin{align}
&\bm{q}(P_{\mathrm{exp}}(H) \psi(0)) \\
=& \bm{q}\left( P_{\mathrm{cos}}(A) \sqrt{M} \dot{x}(0) - P_{\mathrm{sin},v}(A) \sqrt{M} x(0) \right) + \bm{q}\left( iB^\dag (P_{\mathrm{sin}}(A) \sqrt{M} \dot{x}(0) + P_{\mathrm{cos},x}(A) \sqrt{M} x(0)) \right) \\
=& O\left(d_{\mathrm{exp}}^2 \mathcal{N}(d_{\mathrm{exp}} r_0)\mathcal{N}(r_0) \left( \bm{q}(K)\mathcal{N}(r_0) + 2\bm{q}(m) \right) + d_{\mathrm{exp}} \mathcal{N}(d_{\mathrm{exp}} r_0) \left( \bm{q}(\dot{x}(0)) + \bm{q}(x(0)) + 2\bm{q}(m) \right) \right).
\end{align}

From Lemma~\ref{lem:sample-query-est}, for $\zeta \leq \varepsilon/18$ we can estimate $v^\dag P_{\mathrm{exp}}(H) \psi(0)$ classically with a probability at least $1-\delta$ and an additive error $\varepsilon/2$ in
\begin{align}
& O\left( \frac{\log{(1/\delta)}}{\varepsilon^2}\left(\bm{q}(P_{\mathrm{exp}}(H) \psi(0))+\bm{sq}(v)\right) \right) \\
\begin{split}
=& O\left( \frac{\log{(1/\delta)}}{\varepsilon^2}\left( d_{\mathrm{exp}}^2 \mathcal{N}(d_{\mathrm{exp}} r_0)\mathcal{N}(r_0) \left( \bm{q}(K)\mathcal{N}(r_0) + 2\bm{q}(m) \right) \right. \right. \\
& \hspace{4cm} \left. \vphantom{\frac{\log{(1/\delta)}}{\varepsilon^2}} \left. \vphantom{d_{\mathrm{exp}}^2}
+ d_{\mathrm{exp}} \mathcal{N}(d_{\mathrm{exp}} r_0) \left( \bm{q}(\dot{x}(0)) + \bm{q}(x(0)) + 2\bm{q}(m) \right) +\bm{sq}(v)\right) \right)
\end{split}
\end{align}
-time. 
Emphasizing the dependence on $t$, $\varepsilon$ and $\delta$, we can write the complexity
\begin{equation}
\tilde{O}\left( \frac{t^{2+D}\log{(1/\delta)}}{\varepsilon^2} \right).
\end{equation}
\end{proof}

Inspired by Ref.~\cite{babbush2023exponential}, we define another problem, where we calculate kinetic and potential energy.
\begin{problem}\label{problem:energy}
Let $t\in \mathbb{R}$ be an evolution time.
Suppose that we are given
\begin{itemize}
\item query-access to $\dot{x}(0)\in \mathbb{R}^N$,
\item query-access to $x(0)\in \mathbb{R}^N$,
\item query-access to $m\in \mathbb{R}^N$, where $m_i$ is the mass of $i$-th point mass for all $i\in\{1,\dots,N\}$,
\item query-access to $K\in \mathbb{R}^{N\times N}$, where $K_{ij}=\kappa_{ij}$, 
\item total energy: $E= \frac{1}{2} \sum_{i} m_i \dot{x}_i(0)^2 + \frac{1}{2} \sum_{i}\kappa_{ii} x_i(0)^2 + \frac{1}{2} \sum_{j> i}\kappa_{ij} (x_i(0)-x_j(0))^2$,
\item a subsets of masses $\mathcal{V}\subseteq \{1,\dots,N\}$ and a subset of springs $\mathcal{X}\subseteq \{1,\dots N(N+1)/2\}$,
\item query-access to $V=V_{\mathcal{V}}+V_{\mathcal{X}}=\sum_{i\in \mathcal{V}}e_i e_i^\dag + \sum_{i\in \mathcal{X}}e_i e_i^\dag \in \mathbb{R}^{(N+N(N+1)/2)\times (N+N(N+1)/2)}$, where $e_i$ is a unit vector with an 1 in the $i$-th position and zeros elsewhere,
\item $\zeta$-sampling-and-query-access to $\psi(0) = \frac{1}{\sqrt{2E}} ( (\sqrt{M}\dot{x}(0))^T, (i B^\dag \sqrt{M} x(0))^T )^T \in \mathbb{C}^{N+N(N+1)/2}$.
\end{itemize}
Then the goal is to estimate $\hat{z}\in\mathbb{C}$ such that
\begin{equation}
\left|\hat{z} - \frac{K_{\mathcal{V}}(t) + U_{\mathcal{X}}(t)}{E} \right| \leq \varepsilon,
\end{equation}
where
\begin{equation}
K_{\mathcal{V}}(t) \coloneqq \frac{1}{2} \sum_{i\in \mathcal{V}}m_i \dot{x}_i(t)^2,
\end{equation} 
is the kinetic energy of the masses in $\mathcal{V}$ at time $t$ and
\begin{equation}
U_{\mathcal{X}}(t) \coloneqq \frac{1}{2} \sum_{i:(i,i)\in \mathcal{X}}\kappa_{ii} x_i(t)^2 + \frac{1}{2} \sum_{j> i: (i,j)\in \mathcal{X}}\kappa_{ij} (x_i(t)-x_j(t))^2,
\end{equation} 
is the potential energy of the springs in $\mathcal{X}$ at time $t$.
\end{problem}
\begin{cor}\label{cor:energy}
For $\zeta \leq \varepsilon/27$, we can solve Problem~\ref{problem:energy} classically with a probability at least $1-\delta$.
The runtime of the classical algorithm is
\begin{equation}
\begin{split}
O&\left( \frac{\log{(1/\delta)}}{\varepsilon^2}\left( d_{\mathrm{exp}}^2 \mathcal{N}(d_{\mathrm{exp}} r_0)\mathcal{N}(r_0) \bm{q}(A) \right. \right. \\
&\hspace{1cm} \left. \vphantom{\frac{\log{(1/\delta)}}{\varepsilon^2}} \left. \vphantom{d_{\mathrm{exp}}^2} 
+ d_{\mathrm{exp}} \mathcal{N}(d_{\mathrm{exp}} r_0) \mathcal{N}(r_0)\left(3\bm{q}(K) + 3\bm{q}(m) + \bm{q}(V) + 2\bm{q} \left( \phi_x(t) \right)\right) +\bm{sq}(\psi(0))\right) \right),
\end{split}
\end{equation}
where $d_{\mathrm{exp}} \coloneqq O( t \sqrt{\mathcal{N}(r_0)\kappa_{\mathrm{max}}/m_{\mathrm{min}}} + \log(1/\varepsilon))$ and
\begin{equation}
\begin{split}
&\bm{q}\left( \phi_x(t) \right) = O\left(d_{\mathrm{exp}}^2 \mathcal{N}(d_{\mathrm{exp}} r_0)\mathcal{N}(r_0) \left( \bm{q}(K)\mathcal{N}(r_0) + 2\bm{q}(m) \right) \right. \\
&\hspace{5cm} \left. \vphantom{d_{\mathrm{exp}}^2}
+ d_{\mathrm{exp}} \mathcal{N}(d_{\mathrm{exp}} r_0) \left( \bm{q}(\dot{x}(0)) + \bm{q}(x(0)) + 2\bm{q}(m) \right) \right).
\end{split}
\end{equation}
, or emphasizing the dependence on $t$, $\varepsilon$ and $\delta$, given the spatial dimension $D$,
\begin{equation}
\tilde{O}\left( \frac{t^{3+2D}\log{(1/\delta)}}{\varepsilon^2} \right).
\end{equation}
\end{cor}
\begin{proof}
Using some equations used in the proof of Theorem~\ref{thm:classical-time}, the kinetic and potential energy is written as:
\begin{equation}
\frac{K_{\mathcal{V}}(t) + U_{\mathcal{X}}(t)}{E} = \psi(t)^\dag V \psi(t) = \psi(0) e^{-iHt} V e^{iHt} \psi(0) \approx \psi(0) P_{\mathrm{exp}}(H)^\dag V P_{\mathrm{exp}}(H) \psi(0) .
\end{equation}
This implies that, if we have $\varepsilon/3$-approximation $P_{\mathrm{exp}}(H)$ of $e^{iHt}$ and estimate $\psi(0) P_{\mathrm{exp}}(H)^\dag V P_{\mathrm{exp}}(H) \psi(0)$ within an additive error $\varepsilon/3$, then we obtain the overall error $\varepsilon$.
Note that the degree of $P_{\mathrm{exp}}(H)$ is $d_{\mathrm{exp}} \coloneqq O( t \sqrt{\mathcal{N}(r_0)\kappa_{\mathrm{max}}/m_{\mathrm{min}}} + \log(1/\varepsilon))$ as in the proof of Theorem~\ref{thm:classical-time}.

Now we construct a query-access to $P_{\mathrm{exp}}(H)^\dag V P_{\mathrm{exp}}(H) \psi(0)$.
Using the notation used in the proof of Theorem~\ref{thm:classical-time},
\begin{align}
& P_{\mathrm{exp}}(H)^\dag V P_{\mathrm{exp}}(H) \psi(0)\\
=& 
\begin{pmatrix}
   P_{\mathrm{cos}}(A) & -i P_{\mathrm{sin}}(A) B \\
   -i P_{\mathrm{sin}}(B^\dag B) B^\dag & P_{\mathrm{cos}}(B^\dag B)
\end{pmatrix}
\begin{pmatrix}
   V_{\mathcal{V}} & 0 \\
   0 & V_{\mathcal{X}}
\end{pmatrix}
\begin{pmatrix}
   P_{\mathrm{cos}}(A) & i B P_{\mathrm{sin}}(B^\dag B) \\
   i B^\dag P_{\mathrm{sin}}(A) & P_{\mathrm{cos}}(B^\dag B)
\end{pmatrix} 
\psi(0) \\
=& 
\begin{pmatrix}
   P_{\mathrm{cos}}(A) & -i P_{\mathrm{sin}}(A) B \\
   -i P_{\mathrm{sin}}(B^\dag B) B^\dag & P_{\mathrm{cos}}(B^\dag B)
\end{pmatrix}
\begin{pmatrix}
   V_{\mathcal{V}} & 0 \\
   0 & V_{\mathcal{X}} 
\end{pmatrix}
\cdot
\frac{1}{\sqrt{2E}}
\begin{pmatrix}
P_{\mathrm{cos}}(A) \sqrt{M} \dot{x}(0) - P_{\mathrm{sin},v}(A) \sqrt{M} x(0) \\
i B^\dag (P_{\mathrm{sin}}(A) \sqrt{M} \dot{x}(0) + P_{\mathrm{cos},x}(A) \sqrt{M} x(0))
\end{pmatrix} \\
=& 
\begin{pmatrix}
   P_{\mathrm{cos}}(A) & -i P_{\mathrm{sin}}(A) B \\
   -i P_{\mathrm{sin}}(B^\dag B) B^\dag & P_{\mathrm{cos}}(B^\dag B)
\end{pmatrix}
\begin{pmatrix}
   V_{\mathcal{V}} & 0 \\
   0 & V_{\mathcal{X}} B^\dag
\end{pmatrix}
\cdot
\frac{1}{\sqrt{2E}}
\begin{pmatrix}
P_{\mathrm{cos}}(A) \sqrt{M} \dot{x}(0) - P_{\mathrm{sin},v}(A) \sqrt{M} x(0) \\
i (P_{\mathrm{sin}}(A) \sqrt{M} \dot{x}(0) + P_{\mathrm{cos},x}(A) \sqrt{M} x(0))
\end{pmatrix} \\
=& 
\begin{pmatrix}
   P_{\mathrm{cos}}(A) V_{\mathcal{V}} & -i P_{\mathrm{sin}}(A) B V_{\mathcal{X}} B^\dag \\
   -i P_{\mathrm{sin}}(B^\dag B) B^\dag V_{\mathcal{V}} & P_{\mathrm{cos}}(B^\dag B) V_{\mathcal{X}} B^\dag
\end{pmatrix}
\cdot
\begin{pmatrix}
   \phi_v(t) \\
   \phi_x(t)
\end{pmatrix}\\
=& 
\begin{pmatrix}
   P_{\mathrm{cos}}(A) V_{\mathcal{V}} \phi_v(t) -i P_{\mathrm{sin}}(A) B V_{\mathcal{X}} B^\dag \phi_x(t) \\
   -i P_{\mathrm{sin}}(B^\dag B) B^\dag V_{\mathcal{V}} \phi_v(t) + P_{\mathrm{cos}}(B^\dag B) V_{\mathcal{X}} B^\dag \phi_x(t)
\end{pmatrix},\label{eq:PVPpsi}
\end{align}
where $\phi_v(t)\coloneqq \frac{1}{\sqrt{2E}} (P_{\mathrm{cos}}(A) \sqrt{M} \dot{x}(0) - P_{\mathrm{sin},v}(A) \sqrt{M} x(0) )$ and $\phi_x(t) \coloneqq \frac{i}{\sqrt{2E}} (P_{\mathrm{sin}}(A) \sqrt{M} \dot{x}(0) + P_{\mathrm{cos},x}(A) \sqrt{M} x(0))$.
In the proof of Theorem~\ref{thm:classical-time}, we have shown that 
\begin{equation}
\begin{split}
    & \bm{q}\left( \phi_v(t) \right) = \bm{q}\left( \phi_x(t) \right) \\
=& O\left(d_{\mathrm{exp}}^2 \mathcal{N}(d_{\mathrm{exp}} r_0)\mathcal{N}(r_0) \left( \bm{q}(K)\mathcal{N}(r_0) + 2\bm{q}(m) \right) + d_{\mathrm{exp}} \mathcal{N}(d_{\mathrm{exp}} r_0) \left( \bm{q}(\dot{x}(0)) + \bm{q}(x(0)) + 2\bm{q}(m) \right) \right).
\end{split}
\end{equation}
From now, we consider each term in Eq.~\eqref{eq:PVPpsi} respectively.
The $i$-th entry of $V_{\mathcal{V}} \phi_v(t)$ can be calculated by querying the $(i,i)$-th entry of $V_{\mathcal{V}}$ and the $i$-th entry of $\phi_v(t)$ since $V_{\mathcal{V}}$ is a projector and a diagonal matrix.
Thus $\bm{q}(V_{\mathcal{V}} \phi_v(t)) = \bm{q}(V) + \bm{q}(\phi_v(t))$.
Combined with Lemma~\ref{lem:query-GLI-poly}, we obtain
\begin{equation}\label{eq:cor-energy-ul}
\begin{split}
&\bm{q}(P_{\mathrm{cos}}(A) V_{\mathcal{V}} \phi_v(t)) \\
=& O\left(d_{\mathrm{exp}}^2 \mathcal{N}(d_{\mathrm{exp}} r_0)\mathcal{N}(r_0) \bm{q}(A) + d_{\mathrm{exp}} \mathcal{N}(d_{\mathrm{exp}} r_0) \left( \bm{q}(V) + \bm{q}(\phi_v(t)) \right) \right).
\end{split}
\end{equation}
Next, we consider $P_{\mathrm{sin}}(B^\dag B) B^\dag V_{\mathcal{V}} \phi_v(t)$.
We define $B^\dag P_{\mathrm{sin},x}(A) \coloneqq P_{\mathrm{sin}}(B^\dag B) B^\dag$, which is well-defined since $P_{\mathrm{sin}}$ is an odd polynomial.
It is easy to see that $\bm{q}(P_{\mathrm{sin},x}(A) V_{\mathcal{V}} \phi_v(t) ) = O( \bm{q}(P_{\mathrm{cos}}(A) V_{\mathcal{V}} \phi_v(t)) )$.
In the proof of Theorem~\ref{thm:classical-time}, we have proven that, for any vector $u$,
\begin{equation}\label{eq:query-Bdagu}
    \bm{q}(B^\dag u ) = 2\bm{q}(K) + 2\bm{q}(m) + 2\bm{q}(u),
\end{equation}
which implies that
\begin{equation}\label{eq:cor-energy-dl}
\begin{split}
    &\bm{q}(P_{\mathrm{sin}}(B^\dag B) B^\dag V_{\mathcal{V}} \phi_v(t) ) \\
    =& \bm{q}(B^\dag P_{\mathrm{sin},x}(A) V_{\mathcal{V}} \phi_v(t) ) \\
    =& 2\bm{q}(K) + 2\bm{q}(m) + 2\bm{q}(P_{\mathrm{sin},x}(A) V_{\mathcal{V}} \phi_v(t)) \\
    =& 2\bm{q}(K) + 2\bm{q}(m) + O\left(d_{\mathrm{exp}}^2 \mathcal{N}(d_{\mathrm{exp}} r_0)\mathcal{N}(r_0) \bm{q}(A) + d_{\mathrm{exp}} \mathcal{N}(d_{\mathrm{exp}} r_0) \left( \bm{q}(V) + \bm{q}(\phi_v(t)) \right) \right).
\end{split}
\end{equation}
In order to analyze $P_{\mathrm{sin}}(A) B V_{\mathcal{X}} B^\dag \phi_x(t)$, we compute $\bm{q}(B u)$ for any vector $u$.
Observing that
\begin{equation}
\bra{i} B = \sum_{j\geq i}\sqrt{\frac{\kappa_{ij}}{m_i}} \bra{i}\bra{j} - \sum_{j< i}\sqrt{\frac{\kappa_{ij}}{m_i}} \bra{j}\bra{i},
\end{equation}
for the $i$-th row of $B$, there are at most $\mathcal{N}(r_0)$ non-zero entries and each element has the form of $\sqrt{\kappa_{ij}/m_i}$.
Querying all non-zero entries of the $i$-th row of $B$ and their positions and computing corresponding $\mathcal{N}(r_0)$ elements of $u$, we can calculate the $i$-th entry of $B u$.
Thus
\begin{equation}\label{eq:query-Bu}
    \bm{q}(B u) = \mathcal{N}(r_0) \left( \bm{q}(K) + \bm{q}(m) + \bm{q}(u) \right),
\end{equation}
holds.
Using Eq.~\eqref{eq:query-Bdagu} and Lemma~\ref{lem:query-GLI-poly}, we obtain
\begin{align}
    &\bm{q}\left( P_{\mathrm{sin}}(A) B V_{\mathcal{X}} B^\dag \phi_x(t) \right) \\
    =& O\left(d_{\mathrm{exp}}^2 \mathcal{N}(d_{\mathrm{exp}} r_0)\mathcal{N}(r_0) \bm{q}(A) + d_{\mathrm{exp}} \mathcal{N}(d_{\mathrm{exp}} r_0) \bm{q} \left( B V_{\mathcal{X}} B^\dag \phi_x(t) \right) \right) \\
    =& O\left(d_{\mathrm{exp}}^2 \mathcal{N}(d_{\mathrm{exp}} r_0)\mathcal{N}(r_0) \bm{q}(A) + d_{\mathrm{exp}} \mathcal{N}(d_{\mathrm{exp}} r_0) \mathcal{N}(r_0)\left(\bm{q}(K) + \bm{q}(m) + \bm{q} \left( V_{\mathcal{X}} B^\dag \phi_x(t) \right)\right) \right) \\
    =& O\left(d_{\mathrm{exp}}^2 \mathcal{N}(d_{\mathrm{exp}} r_0)\mathcal{N}(r_0) \bm{q}(A) + d_{\mathrm{exp}} \mathcal{N}(d_{\mathrm{exp}} r_0) \mathcal{N}(r_0)\left(3\bm{q}(K) + 3\bm{q}(m) + \bm{q}(V) + 2\bm{q} \left( \phi_x(t) \right)\right) \right) . \label{eq:cor-energy-ur}
\end{align}
Finally, we consider $P_{\mathrm{cos}}(B^\dag B) V_{\mathcal{X}} B^\dag \phi_x(t)$.
We decompose $P_{\mathrm{cos}}(B^\dag B)$ into the identity term and another term such that $a_0 I + P_{\mathrm{cos},i}(B^\dag B) \coloneqq P_{\mathrm{cos}}(B^\dag B)$.
Furthermore, we define $B^\dag P_{\mathrm{cos},ix}(A) B \coloneqq P_{\mathrm{cos},i}(B^\dag B)$, which is well-defined since $P_{\mathrm{cos},i}(B^\dag B)$ does not have an identity term.
Observing that
\begin{equation}
\begin{split}
    &P_{\mathrm{cos}}(B^\dag B) V_{\mathcal{X}} B^\dag \phi_x(t) \\
    =& a_0 V_{\mathcal{X}} B^\dag \phi_x(t) + P_{\mathrm{cos},i}(B^\dag B) V_{\mathcal{X}} B^\dag \phi_x(t) \\
    =& a_0 V_{\mathcal{X}} B^\dag \phi_x(t) + B^\dag P_{\mathrm{cos},ix}(A) B V_{\mathcal{X}} B^\dag \phi_x(t),
\end{split}
\end{equation}
from Eqs.~\eqref{eq:query-Bdagu},~\eqref{eq:query-Bu} and Lemma~\ref{lem:query-GLI-poly}, we obtain
\begin{align}
    & \bm{q} \left( P_{\mathrm{cos}}(B^\dag B) V_{\mathcal{X}} B^\dag \phi_x(t) \right)\\
    =& \bm{q} \left( a_0 V_{\mathcal{X}} B^\dag \phi_x(t) \right) + \bm{q} \left( B^\dag P_{\mathrm{cos},ix}(A) B V_{\mathcal{X}} B^\dag \phi_x(t) \right)\\
    =& \bm{q}(V) + \bm{q}\left( B^\dag \phi_x(t) \right) + 2\bm{q}(K) + 2\bm{q}(m) + 2\bm{q}\left( P_{\mathrm{cos},ix}(A) B V_{\mathcal{X}} B^\dag \phi_x(t) \right) \\
    \begin{split}
        =& 4\bm{q}(K) + 4\bm{q}(m) + \bm{q}(V) + 2\bm{q}\left( \phi_x(t) \right) \\
        &+ O\left(d_{\mathrm{exp}}^2 \mathcal{N}(d_{\mathrm{exp}} r_0)\mathcal{N}(r_0) \bm{q}(A) + d_{\mathrm{exp}} \mathcal{N}(d_{\mathrm{exp}} r_0) \mathcal{N}(r_0)\left(3\bm{q}(K) + 3\bm{q}(m) + \bm{q}(V) + 2\bm{q} \left( \phi_x(t) \right)\right) \right)
    \end{split}\\
    =& O\left(d_{\mathrm{exp}}^2 \mathcal{N}(d_{\mathrm{exp}} r_0)\mathcal{N}(r_0) \bm{q}(A) + d_{\mathrm{exp}} \mathcal{N}(d_{\mathrm{exp}} r_0) \mathcal{N}(r_0)\left(3\bm{q}(K) + 3\bm{q}(m) + \bm{q}(V) + 2\bm{q} \left( \phi_x(t) \right)\right) \right). \label{eq:cor-energy-dr}
\end{align}
Combining Eqs.~\eqref{eq:cor-energy-ul},~\eqref{eq:cor-energy-dl},~\eqref{eq:cor-energy-ur} and~\eqref{eq:cor-energy-dr}, we get 
\begin{align}
    & \bm{q}\left( P_{\mathrm{exp}}(H)^\dag V P_{\mathrm{exp}}(H) \psi(0) \right)\\
    \begin{split}
    =& \bm{q}(P_{\mathrm{cos}}(A) V_{\mathcal{V}} \phi_v(t)) + \bm{q}(P_{\mathrm{sin}}(B^\dag B) B^\dag V_{\mathcal{V}} \phi_v(t) ) \\
    &\hspace{3cm} + \bm{q}\left( P_{\mathrm{sin}}(A) B V_{\mathcal{X}} B^\dag \phi_x(t) \right) + \bm{q} \left( P_{\mathrm{cos}}(B^\dag B) V_{\mathcal{X}} B^\dag \phi_x(t) \right) 
    \end{split}\\
    =& O\left(d_{\mathrm{exp}}^2 \mathcal{N}(d_{\mathrm{exp}} r_0)\mathcal{N}(r_0) \bm{q}(A) + d_{\mathrm{exp}} \mathcal{N}(d_{\mathrm{exp}} r_0) \mathcal{N}(r_0)\left(3\bm{q}(K) + 3\bm{q}(m) + \bm{q}(V) + 2\bm{q} \left( \phi_x(t) \right)\right) \right),
\end{align}
where we used $\bm{q}(\phi_v(t))=O(\bm{q}(\phi_x(t)))$ in the last equality.

From Lemma~\ref{lem:sample-query-est}, for $\zeta \leq \varepsilon/27$ we can estimate $\psi(0) P_{\mathrm{exp}}(H)^\dag V P_{\mathrm{exp}}(H) \psi(0)$ classically with a probability at least $1-\delta$ and an additive error $\varepsilon/3$ in
\begin{align}
& O\left( \frac{\log{(1/\delta)}}{\varepsilon^2}\left(\bm{q}(P_{\mathrm{exp}}(H)^\dag V P_{\mathrm{exp}}(H) \psi(0))+\bm{sq}(\psi(0))\right) \right) \\
\begin{split}
=& O\left( \frac{\log{(1/\delta)}}{\varepsilon^2}\left( d_{\mathrm{exp}}^2 \mathcal{N}(d_{\mathrm{exp}} r_0)\mathcal{N}(r_0) \bm{q}(A) \right. \right. \\
&\hspace{1cm} \left. \vphantom{\frac{\log{(1/\delta)}}{\varepsilon^2}} \left. \vphantom{d_{\mathrm{exp}}^2}
+ d_{\mathrm{exp}} \mathcal{N}(d_{\mathrm{exp}} r_0) \mathcal{N}(r_0)\left(3\bm{q}(K) + 3\bm{q}(m) + \bm{q}(V) + 2\bm{q} \left( \phi_x(t) \right)\right) +\bm{sq}(\psi(0))\right) \right)
\end{split}
\end{align}
-time, where
\begin{equation}
\begin{split}
\bm{q}\left( \phi_x(t) \right) 
&= O\left(d_{\mathrm{exp}}^2 \mathcal{N}(d_{\mathrm{exp}} r_0)\mathcal{N}(r_0) \left( \bm{q}(K)\mathcal{N}(r_0) + 2\bm{q}(m) \right) \right. \\
&\hspace{3cm} \left. \vphantom{d_{\mathrm{exp}}^2}
+ d_{\mathrm{exp}} \mathcal{N}(d_{\mathrm{exp}} r_0) \left( \bm{q}(\dot{x}(0)) + \bm{q}(x(0)) + 2\bm{q}(m) \right) \right).
\end{split}
\end{equation}
Emphasizing the dependence on $t$, $\varepsilon$ and $\delta$, we can write the complexity
\begin{equation}
\tilde{O}\left( \frac{t^{3+2D}\log{(1/\delta)}}{\varepsilon^2} \right).
\end{equation}
\end{proof}

According to Corollary~\ref{cor:energy}, classical computers can calculate kinetic and potential energy in $O(d_{\mathrm{exp}}^{3+2D}/\varepsilon^2)$-time while the runtime of the corresponding quantum algorithm is $O(d_{\mathrm{exp}}/\varepsilon)$~\cite{babbush2023exponential}.
This is one of the demonstrations that the time complexity of classical algorithms differs from that of quantum algorithms only polynomially.

\section{Computational complexity of simulating short-time dynamics under the geometrically local matrix}\label{sec:complexity-short-time-GLI}
In this section, we show that the computational complexity of simulating short-time (i.e. $t=\mathrm{polylog}(N)$) dynamics under the geometrically local matrix is upper-bounded by that of $\mathrm{polylog}(N)$-time probabilistic classical computation and is lower-bounded by that of $\mathrm{polylog}(N)$-time and $O(n)$-space probabilistic classical computation.
To this end, we first show that $\mathrm{polylog}(N)$-time probabilistic classical computers can simulate the short-time dynamics.
In Sec.~\ref{sec:dequantizing-QEVT} and Sec.~\ref{sec:applications}, we have provided the dequantized algorithm for estimating the quantity given by this short-time dynamics.
However, this is insufficient as a result of classical easiness since it does not address the sampling problem, which is considered a more challenging task in the context of quantum advantage.
In the following, we thus show that sampling from the short-time evolved state under the geometrically local matrix can also be simulated classically.

\subsection{Classical sampling simulation of short-time dynamics under the geometrically local matrix}
Here we provide the $\mathrm{polylog}(N)$-time classical algorithm that simulates the sampling from the quantum state which evolved for a short time under the geometrically local matrix of size $N\times N$.
In the following, we assume that the time evolution operator $e^{At}$ is approximated by a polynomial transformation of $A$ such that $\|P_{\mathrm{exp}}(A) - e^{At}\| \leq \varepsilon$, where the degree of $P_{\mathrm{exp}}$ is at most $d_{\mathrm{exp}}=O(\mathrm{poly}(t, n,1/\varepsilon))$.
When $A$ is anti-Hermitian, we can set the degree as $d_{\mathrm{exp}}=O(\|A\| t + \log{(1/\varepsilon)})$ as in Ref.~\cite{gilyen2019quantum}.
The following lemma ensures that the state approximating the time evolved state by polynomial transformation also approximates the output probability distribution in the total variation distance.

\begin{lem}\label{lem:operator-error-to-total-variation}
Let $t\in \mathbb{R}$ be an evolution time and $\varepsilon\in(0,1]$ be a precision.
Suppose that $A$ is a square matrix and $\ket{\psi}$ is an initial state.
Assume that $P_{\mathrm{exp}}$ is a polynomial of degree $d_{\mathrm{exp}}$ that $P_{\mathrm{exp}}(A)$ $\varepsilon$-approximates $e^{At}$ such that $\|P_{\mathrm{exp}}(A) - e^{At}\| \leq \varepsilon$.
Let $p^{\mathrm{app}}\coloneqq p^{P_{\mathrm{exp}}(A)\ket{\psi}}$ and $p^{\mathrm{exact}}\coloneqq p^{e^{At}\ket{\psi}}$ be the output distributions of $P_{\mathrm{exp}}(A)\ket{\psi}$ and $e^{At} \ket{\psi}$.
Then the total variation distance between $p^{\mathrm{app}}$ and $p^{\mathrm{exact}}$ is upper-bounded by
\begin{equation}
|p^{\mathrm{app}}-p^{\mathrm{exact}}|_{\mathrm{tv}} \leq \frac{2\varepsilon}{\| e^{At} \ket{\psi} \|}.
\end{equation}
\end{lem}
\begin{proof}
Let $\ket{P_{\mathrm{exp}}(A)\psi}_{\mathrm{N}}\coloneqq P_{\mathrm{exp}}(A)\ket{\psi}/\|P_{\mathrm{exp}}(A)\ket{\psi}\|$ and $\ket{e^{At}\psi}_{\mathrm{N}}\coloneqq e^{At}\ket{\psi}/\|e^{At}\ket{\psi}\|$ denote the normalized state and note that $p^{\ket{P_{\mathrm{exp}}(A)\psi}_{\mathrm{N}}} = p^{P_{\mathrm{exp}}(A)\ket{\psi}}=p^{\mathrm{app}}$ and $p^{\ket{e^{At}\psi}_{\mathrm{N}}} = p^{e^{At}\ket{\psi}}=p^{\mathrm{exact}}$.

First, we show that $\|\ket{P_{\mathrm{exp}}(A)\psi}_{\mathrm{N}} - \ket{e^{At}\psi}_{\mathrm{N}}\|\leq 2\varepsilon/ \| e^{At} \ket{\psi} \|$.
For any vectors $u$ and $\phi$ such that $\|\phi\|=1$, we have
\begin{equation}
\left\| \frac{u}{\|u\|} - \phi \right\| 
\leq \left\| \frac{u}{\|u\|} - u \right\| + \left\| u - \phi \right\|
\leq 2\| u - \phi \|.
\end{equation}
Substituting $u=P_{\mathrm{exp}}(A)\ket{\psi}/\| e^{At} \ket{\psi} \|$ and $\phi = e^{At} \ket{\psi}/\| e^{At} \ket{\psi} \|$ into the above inequality, we obtain
\begin{align}
\|\ket{P_{\mathrm{exp}}(A)\psi}_{\mathrm{N}} - \ket{e^{At}\psi}_{\mathrm{N}}\|
&=  \left\| \frac{\| e^{At} \ket{\psi} \|}{\|P_{\mathrm{exp}}(A)\ket{\psi}\|} \frac{P_{\mathrm{exp}}(A)\ket{\psi}}{\| e^{At} \ket{\psi} \|}   - \frac{e^{At}\ket{\psi}}{\| e^{At} \ket{\psi} \|} \right\| \\
&\leq 2 \frac{1}{\| e^{At} \ket{\psi} \|} \left\| P_{\mathrm{exp}}(A)\ket{\psi}  - e^{At}\ket{\psi} \right\| \\
&\leq 2 \frac{\varepsilon}{\| e^{At} \ket{\psi} \|}.
\end{align}

Recall the well-known fact that, for two quantum states $\ket{\Psi}$ and $\ket{\Phi}$, when $\| \ket{\Psi} - \ket{\Phi}\| \leq \varepsilon$, we have $| p^{\ket{\Psi}} - p^{\ket{\Phi}} |_{\mathrm{tv}}\leq \varepsilon$.
Then we obtain
\begin{equation}
|p^{\mathrm{app}}-p^{\mathrm{exact}}|_{\mathrm{tv}} \leq \frac{2\varepsilon}{\| e^{At} \ket{\psi} \|}.
\end{equation}
\end{proof}

From the above lemma, we only have to sample from the probability distribution $p^{P_{\mathrm{exp}}(A)\ket{\psi}}$ of $P_{\mathrm{exp}}(A)\ket{\psi}$.
To this end, we use rejection sampling which is a method to sample from the hard distribution $P_{\mathrm{exp}}(A)\ket{\psi}$ given the sampling-access to some easy distribution $p^{\mathrm{over}}$.
In the standard rejection sampling scheme, we need a query-access to the normalized target state $\ket{P_{\mathrm{exp}}(A)\psi}_{\mathrm{N}}$.
In our case, however, we have a query-access to the unnormalized state $P_{\mathrm{exp}}(A)\ket{\psi}$ but $\ket{P_{\mathrm{exp}}(A)\psi}_{\mathrm{N}}$.
Thus we have to show that rejection sampling works with a query-access to the unnormalized state, i.e., without the exact knowledge of the normalization factor.
The following lemma provides the robust rejection sampling with a query-access to the unnormalized state, which is the slightly modified version of Proposition 4 in Ref.~\cite{gall2023robust}.
Note that, from the triangle inequality, we have $\|e^{At}\ket{\psi}\|-\varepsilon \leq \| P_{\mathrm{exp}}(A)\ket{\psi} \| \leq \|e^{At} \ket{\psi} \| + \varepsilon$.

\begin{lem}[Robust rejection sampling with a query-access to the unnormalized state]\label{lem:rejection-sampling-unnormalized}
Let $p\in \mathbb{R}^N$ be a probability distribution and $u\in \mathbb{C}^N$ is a vector such that $\alpha_{\mathrm{min}}-\varepsilon\leq \|u\| \leq \alpha_{\mathrm{max}} + \varepsilon$ for some known $\alpha_{\mathrm{min}}$ and $\alpha_{\mathrm{max}}$.
Suppose that we are given
\begin{itemize}
    \item $\zeta$-sampling-and-query-access to $p$,
    \item query-access to $u$,
    \item a real number $\phi\geq 1$ such that $\frac{|u_i|^2}{\phi p_i}\leq 1$ holds for all $i\in\{1,\dots,N\}$.
\end{itemize}
Assume that $0<\varepsilon \leq \frac{\alpha_{\mathrm{min}}}{2}$ and $0\leq \zeta\leq \frac{\alpha_{\mathrm{min}}^2}{16\phi}$.
Then we can sample from the probability distribution $p'$ such that $|p' - p^u |_{\mathrm{tv}}\leq 16\phi\zeta/\alpha_{\mathrm{min}}^2$ with a probability $\geq 1-\delta$ in $O((\bm{sq}(p) + \bm{q}(u)) \frac{\phi}{\alpha_{\mathrm{min}}^2} \log{(1/\delta)})$-time.
\end{lem}
\begin{proof}
Our sampling procedure is as follows:
\begin{enumerate}
    \item Take a sample $i\in\{1,\dots,N\}$ from the probability distribution $\tilde{p}$ such that $|p-\tilde{p}|_{\mathrm{tv}}\leq \zeta$ using $\zeta$-sampling-access to $p$.
    \item Take $t$ uniformly at random in $[0,1]$.
    \item If $t\leq \frac{|u_i|^2}{\phi p_i}$, then output $i$, otherwise output \textit{failure}.
\end{enumerate}

First, we evaluate the success probability $p_{\mathrm{acc}}$ of this procedure.
We have 
\begin{equation}
p_{\mathrm{acc}} 
= \sum_{i=1}^N \tilde{p}_i\frac{|u_i|^2}{\phi p_i} 
= \sum_{i=1}^N p_i\frac{|u_i|^2}{\phi p_i} + \sum_{i=1}^N (\tilde{p}_i - p_i)\frac{|u_i|^2}{\phi p_i}
= \frac{\|u\|^2}{\phi} + \sum_{i=1}^N (\tilde{p}_i - p_i)\frac{|u_i|^2}{\phi p_i},
\end{equation}
and thus, from the assumption $\frac{|u_i|^2}{\phi p_i}\leq 1$ for all $i$,
\begin{equation}
\left| p_{\mathrm{acc}} - \frac{\|u\|^2}{\phi} \right|
= \left| \sum_{i=1}^N (\tilde{p}_i - p_i)\frac{|u_i|^2}{\phi p_i} \right|
\leq \sum_{i=1}^N |\tilde{p}_i - p_i|
= 2 |p-\tilde{p}|_{\mathrm{tv}}
\leq 2 \zeta.
\end{equation}
We obtain
\begin{equation}
\frac{\|u\|^2}{\phi} - 2\zeta \leq p_{\mathrm{acc}} \leq \frac{\|u\|^2}{\phi} + 2\zeta.
\end{equation}
Then we have
\begin{align}
\frac{\phi}{\|u\|^2+2\phi\zeta} \leq &\frac{1}{p_{\mathrm{acc}}} \leq \frac{\phi}{\|u\|^2-2\phi\zeta} \\
\frac{\phi}{\|u\|^2} \left(1-\frac{2\phi\zeta}{\|u\|^2}\right) \leq &\frac{1}{p_{\mathrm{acc}}} \leq \frac{\phi}{\|u\|^2} \left( 1+\frac{4\phi\zeta}{\|u\|^2} \right) 
\end{align}
for $0\leq\zeta\leq \frac{\alpha_{\mathrm{min}}^2}{16\phi} \leq \frac{\|u\|^2}{4\phi}$.

Let $p':\{1,\dots,N\} \rightarrow [0,1]$ be the output probability distribution of the above procedure conditioned on success.
We have
\begin{equation}
p_i' = \frac{\tilde{p}_i|u_i|^2}{p_{\mathrm{acc}} \phi p_i},
\end{equation}
for all $i\in\{1,\dots,N\}$.
We can see that $p'$ is close to the desired distribution $p^u$:
\begin{align}
|p' - p^u |_{\mathrm{tv}}
& = \frac{1}{2} \sum_{i=1}^N \left| \frac{\tilde{p}_i|u_i|^2}{p_{\mathrm{acc}} \phi p_i} - \frac{|u_i|^2}{\|u\|^2} \right| \\
&= \frac{1}{2} \sum_{i=1}^N \left| \frac{p_i|u_i|^2}{p_{\mathrm{acc}} \phi p_i} - \frac{|u_i|^2}{\|u\|^2} \right| + \frac{1}{2} \sum_{i=1}^N \left| \frac{(\tilde{p}_i - p_i)|u_i|^2}{p_{\mathrm{acc}} \phi p_i} \right| \\
&= \frac{1}{2} \sum_{i=1}^N \left| \frac{p_i|u_i|^2}{p_{\mathrm{acc}} \phi p_i} - \frac{|u_i|^2}{\|u\|^2} \right| + \frac{1}{2} \sum_{i=1}^N \left| \frac{(\tilde{p}_i - p_i)|u_i|^2}{p_{\mathrm{acc}} \phi p_i} \right| \\
&\leq \frac{1}{2} \left| \frac{\|u\|^2}{p_{\mathrm{acc}} \phi} - 1 \right| + \frac{1}{p_{\mathrm{acc}}} |\tilde{p} - p |_{\mathrm{tv}} \\
&\leq \frac{2\phi\zeta}{\|u\|^2} + \frac{\phi \zeta}{\|u\|^2} \left( 1+\frac{4\phi\zeta}{\|u\|^2} \right) \\
& \leq \frac{16\phi \zeta}{\alpha_{\mathrm{min}}^2},
\end{align}
for $0<\varepsilon \leq \frac{\alpha_{\mathrm{min}}}{2}$ and $0\leq \zeta\leq \frac{\alpha_{\mathrm{min}}^2}{16\phi}$.

Now we provide the number of repetitions required to obtain the success probability $\geq 1-\delta$.
When we repeat the procedure $s$ times, the failure probability is upper-bounded by
\begin{equation}
\left( 1-p_{\mathrm{acc}} \right)^s 
\leq \left( 1- \frac{\alpha_{\mathrm{min}}^2}{8\phi} \right)^s
\leq e^{-\alpha_{\mathrm{min}}^2 s/8\phi},
\end{equation}
since $p_{\mathrm{acc}} \geq \frac{\|u\|^2}{\phi} - 2\zeta \geq \frac{\alpha_{\mathrm{min}}^2}{4\phi}-\frac{\alpha_{\mathrm{min}}^2}{8\phi}\geq \frac{\alpha_{\mathrm{min}}^2}{8\phi}$.
Thus, if we choose $s=\frac{8\phi}{\alpha_{\mathrm{min}}^2}\log{(1/\delta)}$, the success probability is lower-bounded by $1-\delta$.

The runtime for one trial is given by $2\bm{sq}(p) + \bm{q}(u)$.
Thus, to get one sample, the overall complexity is $O((\bm{sq}(p) + \bm{q}(u)) \frac{\phi}{\alpha_{\mathrm{min}}^2} \log{(1/\delta)})$.
\end{proof}

To perform rejection sampling, we also need a sampling-access to some easy distribution $p^{\mathrm{over}}$ with an appropriate property.
That is, we need an easy distribution which \textit{oversamples} the desired distribution $p^{\ket{P_{\mathrm{exp}}(A)\psi}_{\mathrm{N}}}$ with a small oversampling parameter $\phi$ (defined in Lemma~\ref{lem:rejection-sampling-unnormalized}).
Now we demonstrate how to get this distribution.
Our strategy is just sampling from the initial state $\ket{\psi}$ followed by sampling from the uniform distribution over the region determined by the light cone.

\begin{lem}\label{lem:oversample-GLI}
Suppose that $A\in\mathbb{C}^{N\times N}$ is an $(r_0,\mathcal{N}(r_0))$-geometrically local matrix and we are given
\begin{itemize}
    \item $\zeta$-sampling-and-query-access to a quantum state $\ket{\psi}\in\mathbb{C}^N$.
\end{itemize}
Let $P\in\mathbb{C}[x]$ be a polynomial of degree $d$.
Then we can sample from the probability distribution $p'$ such that $|p' - p^{\mathrm{over}} |_{\mathrm{tv}}\leq \zeta$, where $p^{\mathrm{over}}$ satisfies that $\frac{|(P(A)\ket{\psi})_i|^2}{\phi p_i^{\mathrm{over}}}\leq 1$ for all $i\in\{1,\dots,N\}$ and $\phi = \|P(A)\|^2 \mathcal{N}(d r_0)$, in $\bm{sq}(\psi)$-time.
\end{lem}
\begin{proof}
Our sampling procedure is as follows:
\begin{itemize}
    \item Take a sample $i\in\{1,\dots,N\}$ from the probability distribution $\tilde{p}^{\ket{\psi}}$ such that $|\tilde{p}^{\ket{\psi}} - p^{\ket{\psi}} |_{\mathrm{tv}} \leq \zeta$ using $\zeta$-sampling-access to $\ket{\psi}$.
    \item Take $j$ uniformly at random in $\{ j| d(i,j)\leq d r_0 \}$.
\end{itemize}
This procedure generates the probability distribution $p'$ such that
\begin{equation}
p_i' = \sum_{j:d(i,j)\leq d r_0} \frac{1}{|\{ k| d(j,k)\leq d r_0 \}|} \tilde{p}_j^{\ket{\psi}}.
\end{equation}
When $\zeta=0$, $p'$ coincides with the desired probability distribution $p^{\mathrm{over}}$ such that
\begin{equation}
p_i^{\mathrm{over}} = \sum_{j:d(i,j)\leq d r_0} \frac{1}{|\{ k| d(j,k)\leq d r_0 \}|} p_j^{\ket{\psi}} = \sum_{j:d(i,j)\leq d r_0} \frac{1}{|\{ k| d(j,k)\leq d r_0 \}|} |\psi_j|^2.
\end{equation}

First, we show that $p'$ is close to $p^{\mathrm{over}}$:
\begin{align}
|p' - p^{\mathrm{over}} |_{\mathrm{tv}}
& = \frac{1}{2} \sum_{i=1}^N \left| p_i' - p_i^{\mathrm{over}} \right| \\
& \leq \frac{1}{2} \sum_{i=1}^N \sum_{j:d(i,j)\leq d r_0} \frac{1}{|\{ k| d(j,k)\leq d r_0 \}|} \left| \tilde{p}_j^{\ket{\psi}} - p_j^{\ket{\psi}} \right| \\
& = \frac{1}{2} \sum_{i=1}^N \sum_{j:d(i,j)\leq d r_0} \frac{1}{|\{ k| d(i,k)\leq d r_0 \}|} \left| \tilde{p}_i^{\ket{\psi}} - p_i^{\ket{\psi}} \right| \\
& \leq \frac{1}{2} \sum_{i=1}^N \left| \tilde{p}_i^{\ket{\psi}} - p_i^{\ket{\psi}} \right| \\
& = |\tilde{p}^{\ket{\psi}} - p^{\ket{\psi}} |_{\mathrm{tv}} \\
& \leq \zeta .
\end{align}

Now we show that $p^{\mathrm{over}}$ $\|P(A)\|^2 \mathcal{N}(d r_0)$-oversamples $P(A)\ket{\psi}$.
First, we consider the subnormalized state $\ket{P(A)\psi}_{\mathrm{SN}}=P(A)\ket{\psi}/\|P(A)\|$.
Let $b_i$ denote the $i$-th element of $\ket{P(A)\psi}_{\mathrm{SN}}$.
Due to the locality of $P(A)$ (see Lemma~\ref{lem:GLI-power}), $b_i$ can be written as a linear combination of $\mathcal{N}(d r_0)$ terms of $\ket{\psi}$
\begin{equation}
    b_i = \sum_{j:d(i,j)\leq d r_0} a_j^* \psi_j = a^\dag \widetilde{\psi},
\end{equation}
where $a$ and $\widetilde{\psi}$ are $\mathcal{N}(d r_0)$-dimensional vectors with elements $a_i$ and $\psi_i$.
Note that $a$ satisfies $\|a\| \leq 1$ since $\left\|\frac{P(A)}{\|P(A)\|}\right\| = 1$.
Define the normalized vector $\lambda_1\coloneqq a/\|a\|$ and construct a orthonormal basis $\{\lambda_1,\dots,\lambda_{\mathcal{N}(d r_0)}\}$ of $\mathcal{N}(d r_0)$-dimensional subspace.
Then we can observe that
\begin{equation}
\sum_{j:d(i,j)\leq d r_0} |\psi_j|^2
=\widetilde{\psi}^\dag \widetilde{\psi} = \widetilde{\psi}^\dag \left(\sum_i \lambda_i \lambda_i^\dag \right) \widetilde{\psi}
= \sum_i \left| \lambda_i^\dag \widetilde{\psi} \right|^2
\geq \left| \lambda_1^\dag \widetilde{\psi} \right|^2
=\frac{1}{\|a\|^2} \left| a^\dag \widetilde{\psi} \right|^2
\geq |b_i|^2.
\end{equation}
Combined with the fact that
\begin{equation}
p_i^{\mathrm{over}} = \sum_{j:d(i,j)\leq d r_0} \frac{1}{|\{ k| d(j,k)\leq d r_0 \}|} |\psi_j|^2
\geq \frac{1}{\mathcal{N}(d r_0)} \sum_{j:d(i,j)\leq d r_0} |\psi_j|^2,
\end{equation}
we obtain
\begin{equation}
\frac{|(\ket{P(A)\psi}_{\mathrm{SN}})_i|^2}{\mathcal{N}(d r_0) p_i^{\mathrm{over}}}
\leq \frac{|b_i|^2}{\sum_{j:d(i,j)\leq d r_0} |\psi_j|^2}
\leq 1
\end{equation}
for all $i$.
This leads to
\begin{equation}
\frac{|(P(A)\ket{\psi})_i|^2}{\|P(A)\|^2 \mathcal{N}(d r_0)  p_i^{\mathrm{over}}}
=\frac{|(\ket{P(A)\psi}_{\mathrm{SN}})_i|^2}{\mathcal{N}(d r_0) p_i^{\mathrm{over}}}
\leq 1,
\end{equation}
which implies that $p^{\mathrm{over}}$ $\|P(A)\|^2 \mathcal{N}(d r_0)$-oversamples $P(A)\ket{\psi}$.

It is easy to see that this procedure only requires a cost of $\bm{sq}(\psi)$.
\end{proof}

Now we are ready to sample from the desired distribution $p^{e^{At}\ket{\psi}}$.

\begin{thm}\label{thm:classical-simulation-sampling-1D-GLI}
Let $t\in \mathbb{R}$ be an evolution time.
Suppose that $A\in\mathbb{C}^{N\times N}$ is an $(r_0,\mathcal{N}(r_0))$-geometrically local matrix with spatial locality $r_0\in\mathbb{R}_{+}$, $\ket{\psi}$ is a quantum state and we are given
\begin{itemize}
    \item $\zeta$-sampling-and-query-access to $\ket{\psi}$,
    \item query-access to $A$.
\end{itemize}
Let $\alpha_{\mathrm{min}}\leq \|e^{At}\ket{\psi}\|$ be an lower bound on norm of the solution vector and $\alpha_{\mathrm{exp}}\geq \|e^{At}\|$ be an upper bound on norm of the time evolution operator.
Assume that $P_{\mathrm{exp}}$ is a polynomial of degree $d_{\mathrm{exp}}$ such that $P_{\mathrm{exp}}(A)$ $\varepsilon$-approximates $e^{At}$ such that $\|P_{\mathrm{exp}}(A) - e^{At}\| \leq \varepsilon$, and that $0<\varepsilon \leq \frac{\alpha_{\mathrm{min}}}{2}$ and $0\leq \zeta\leq \frac{\alpha_{\mathrm{min}}^2}{64 \alpha_{\mathrm{exp}}^2 \mathcal{N}(d_{\mathrm{exp}} r_0)}$.
Let $p^{\mathrm{exact}}\coloneqq p^{e^{At}\ket{\psi}}$ be the output distributions of $e^{At} \ket{\psi}$.
Then we can sample from the probability distribution $\widetilde{p}$ such that
\begin{equation}
|\widetilde{p}-p^{\mathrm{exact}}|_{\mathrm{tv}} \leq \frac{64 \alpha_{\mathrm{exp}}^2 \mathcal{N}(d_{\mathrm{exp}} r_0) \zeta}{\alpha_{\mathrm{min}}^2} + \frac{2\varepsilon}{\alpha_{\mathrm{min}}},
\end{equation}
classically with a probability $\geq 1-\delta$ in $O\left(\left(d_{\mathrm{exp}}^2 \mathcal{N}(r_0) \bm{q}(A) + d_{\mathrm{exp}} \bm{sq}(\psi)\right) \mathcal{N}(d_{\mathrm{exp}} r_0)^2  \frac{\alpha_{\mathrm{exp}}^2}{\alpha_{\mathrm{min}}^2} \log{(1/\delta)}\right)$-time.
\end{thm}
\begin{proof}
Let $p^{\mathrm{app}}\coloneqq p^{P_{\mathrm{exp}}(A)\ket{\psi}}$ be the output distributions of $P_{\mathrm{exp}}(A) \ket{\psi}$.
Then, from Lemma~\ref{lem:operator-error-to-total-variation}, we have
\begin{equation}
|p^{\mathrm{app}}-p^{\mathrm{exact}}|_{\mathrm{tv}} \leq \frac{2\varepsilon}{\| e^{At} \ket{\psi} \|}.
\end{equation}
Thus we consider how to sample from $p^{\mathrm{app}}$.

Our sampling procedure is as follows:
\begin{enumerate}
    \item Take a sample $i\in\{1,\dots,N\}$ from the probability distribution $\tilde{p}^{\mathrm{over}}$ such that $|\tilde{p}^{\mathrm{over}} - p^{\mathrm{over}} |_{\mathrm{tv}} \leq \zeta$, where $p^{\mathrm{over}}$ $\|P_{\mathrm{exp}}(A)\|^2 \mathcal{N}(d_{\mathrm{exp}} r_0)$-oversamples $P_{\mathrm{exp}}(A)\ket{\psi}$.
    \item Take $t$ uniformly at random in $[0,1]$.
    \item If $t\leq \frac{|(P_{\mathrm{exp}}(A) \ket{\psi})_i|^2}{4\alpha_{\mathrm{exp}}^2 \mathcal{N}(d_{\mathrm{exp}} r_0) p_i^{\mathrm{over}}}$, then output $i$, otherwise output \textit{failure}.
\end{enumerate}
According to Lemma~\ref{lem:oversample-GLI}, we can sample from $\tilde{p}^{\mathrm{over}}$ in $\bm{sq}(\psi)$-time.
Note that $\tilde{p}^{\mathrm{over}}$ also $4\alpha_{\mathrm{exp}}^2 \mathcal{N}(d_{\mathrm{exp}} r_0)$-oversamples $P_{\mathrm{exp}}(A)\ket{\psi}$ since $\|P_{\mathrm{exp}}(A)\| \leq \|e^{At}\| + \varepsilon \leq 2\alpha_{\mathrm{exp}}$.
Now we have $\zeta$-sampling-and-query-access to $p^{\mathrm{over}}$.
Combined with Lemma~\ref{lem:rejection-sampling-unnormalized} for $p=p^{\mathrm{over}}$, $u=P_{\mathrm{exp}}(A)\ket{\psi}$ and $\phi = 4\alpha_{\mathrm{exp}}^2 \mathcal{N}(d_{\mathrm{exp}} r_0)$, we can sample from the probability distribution $\widetilde{p}$ such that $|\widetilde{p}-p^{\mathrm{app}}|_{\mathrm{tv}} \leq 64 \alpha_{\mathrm{exp}}^2 \mathcal{N}(d_{\mathrm{exp}} r_0) \zeta/\alpha_{\mathrm{min}}^2$ with a probability $\geq 1-\delta$ in $O((\bm{sq}(p^{\mathrm{over}}) + \bm{q}(P_{\mathrm{exp}}(A)\ket{\psi})) \frac{\alpha_{\mathrm{exp}}^2 \mathcal{N}(d_{\mathrm{exp}} r_0)}{\alpha_{\mathrm{min}}^2} \log{(1/\delta)})$-time.
Considering that $\bm{sq}(p^{\mathrm{over}})=\bm{sq}(\psi)$ and 
\begin{equation}
\bm{q}(P_{\mathrm{exp}}(A)\ket{\psi}) = O\left(d_{\mathrm{exp}}^2 \mathcal{N}(d_{\mathrm{exp}} r_0)\mathcal{N}(r_0) \bm{q}(A) + d_{\mathrm{exp}} \mathcal{N}(d_{\mathrm{exp}} r_0) \bm{q}(\psi) \right),
\end{equation}
(see Lemma~\ref{lem:query-GLI-poly}), we obtain the overall complexity
\begin{align}
&O\left((\bm{sq}(\psi) + d_{\mathrm{exp}}^2 \mathcal{N}(d_{\mathrm{exp}} r_0)\mathcal{N}(r_0) \bm{q}(A) + d_{\mathrm{exp}} \mathcal{N}(d_{\mathrm{exp}} r_0) \bm{q}(\psi)) \mathcal{N}(d_{\mathrm{exp}} r_0) \frac{\alpha_{\mathrm{exp}}^2}{\alpha_{\mathrm{min}}^2} \log{(1/\delta)}\right) \\
=& O\left(\left(d_{\mathrm{exp}}^2 \mathcal{N}(r_0) \bm{q}(A) + d_{\mathrm{exp}} \bm{sq}(\psi)\right) \mathcal{N}(d_{\mathrm{exp}} r_0)^2 \frac{\alpha_{\mathrm{exp}}^2}{\alpha_{\mathrm{min}}^2} \log{(1/\delta)}\right) .
\end{align}
\end{proof}

It is noteworthy that the complexity of quantum algorithms for generating the solution state $e^{At} \ket{\psi}/\| e^{At} \ket{\psi} \|$ is of the form
\begin{equation}
    \tilde{O}\left(\frac{\max_{0\leq \tau \leq t}\|e^{A\tau}\|}{\| e^{At} \ket{\psi} \|} d_{\mathrm{exp}} \right),
\end{equation}
see e.g. Refs.~\cite{fang2023time, berry2024quantum, an2023linear, an2023quantum, krovi2023improved, low2024quantum, low2024quantumlinear}.
In this work, we only consider $O(n)$-qubit quantum algorithm that runs in $\mathrm{poly}(t,n)$-time, which means that $\frac{\|e^{At}\|}{\|e^{At}\ket{\psi}\|} \leq \frac{\alpha_{\mathrm{exp}}}{\alpha_{\mathrm{min}}} =\mathrm{poly}(t,n)$ and $d_{\mathrm{exp}}=\mathrm{poly}(t,n)$.
As a result, our classical algorithm runs in $\mathrm{poly}(t,n)$-time and $\mathrm{poly}(t,n)$-space.
Thus the $O(n)$-qubit quantum algorithm for short-time dynamics under the geometrically local matrix can be simulated classically even when we consider the sampling from the solution state.

\subsection{Short-time dynamics of the 1D geometrically local Hamiltonian can simulate classical computation}\label{subsec:1DGLI-simulate-QC-BPgate}
Next we show the opposite direction, that is, the computational complexity of simulating short-time ($\mathrm{polylog}(N)$-time) dynamics under the geometrically local matrix is at least as hard as that of $\mathrm{polylog}(N)$-time and $O(n)$-space probabilistic classical computation.
To this end, we embed the classical reversible circuit to some simple geometrically local Hamiltonian.
In the following, we consider the system of classical coupled harmonic oscillators to show a hardness result for a physically meaningful system.

First, we construct the Hamiltonian based on the standard Feynman-Kitaev circuit-to-Hamiltonian construction~\cite{feynman1986quantum, kitaev2002classical}.
To embed the circuit to the Hamiltonian of classical coupled harmonic oscillators, we use the following construction:
\begin{equation}\label{eq:FK-classical-reversible}
A_{\mathrm{FK}} \coloneqq 2I - \sum_{l=1}^L \left(  \ket{l+1}\bra{l} \otimes U_l + \ket{l}\bra{l+1} \otimes U_l^\dag \right),
\end{equation}
where $U = U_{L} \dots U_1$ is a $n$-qubit unitary operator generated by $L=\mathrm{poly}(n)=\mathrm{polylog}(N)$ classical reversible gates, e.g., $\{ X, \mathrm{CNOT}, \mathrm{Toffoli} \}$.
When $U_1,\dots,U_L$ are classical reversible gates, the above Hamiltonian $A_{\mathrm{FK}}\in \mathbb{C}^{(L+1)2^n \times (L+1)2^n}$ can be viewed as $(1,\mathcal{N}(1))$-geometrically local Hamiltonian on one-dimensional lattice with $\mathcal{N}(r)=2r+1$.
Intuitively, this Hamiltonian can be understood as tracking the classical computation for all $2^n$ possible inputs, and therefore it can be viewed as a collection of $2^n$ independent one-dimensional geometrically local Hamiltonians of length $L+1$.
We can see this formally as follows:
We consider the two-dimensional lattice of size $2^n \times (L+1)$, whose sites correspond to some computational basis state, as shown in Fig.~\ref{fig:1d_system_classical_reversible_circuit}.
The indices of this lattice is determined as follows:
The sites in the $l=1$-th column represent the input computational basis states. 
For simplicity, we assign them from top to bottom as $\ket{0},\ket{1},\dots,\ket{2^n-1}$, and label them the indices $(0,1),(1,1), \dots ,(2^n-1, 1)$.
The sites in the $l=2$-th column represent the output computational basis states after $U_1$ is applied.
That is, if $\ket{f_1(i)} \coloneqq U_1\ket{i}$, then the sites named $(i,1)$ and $(f_1(i),2)$ are connected.
In this way, we can determine all indices in the $l=2$-th column uniquely since $U_1$ is a classical reversible gate (i.e., basis-preserving).
In the same manner, we can also uniquely assign the indices for the sites in the columns $l=2,\dots , L+1$.
For example, the site in the $k$-th row and the $l$-th column has the index $(f_{l-1}\circ\dots \circ f_1(k),l)$, where $\ket{f_{l-1}\circ\dots \circ f_1(k)}\coloneqq U_{l-1}\dots U_1 \ket{k}$, and corresponds to the computational basis state $\ket{l}\ket{f_{l-1}\circ\dots \circ f_1(k)}$.
From the above discussion, we can map the interactions of the Hamiltonian $A_{\mathrm{FK}}$ onto the two-dimensional lattice in Fig.~\ref{fig:1d_system_classical_reversible_circuit} with $m_i=1$, $\kappa_{ii}=0$, and $\kappa_{ij}\in \{0,1\}$ ($i\neq j$), which is inherently one-dimensional and has nearest neighbor interactions $\mathcal{N}(r)=2r+1$.
It should be noted that this argument is solely intended to demonstrate geometric locality, and that no explicit classical computation is required when the quantum algorithm simulates the corresponding dynamics.
Also, the Feynman–Kitaev Hamiltonian may appear to be geometrically local even for a universal gate set since that Hamiltonian can be regarded as a one-dimensional tight-binding model on the clock register.
However, this does not hold, since geometric locality in classical systems is defined with respect to the computational basis state, and the dequantized algorithm use the sampling in this basis.

\begin{figure}
\centerline{
\includegraphics[width=120mm, page=1]{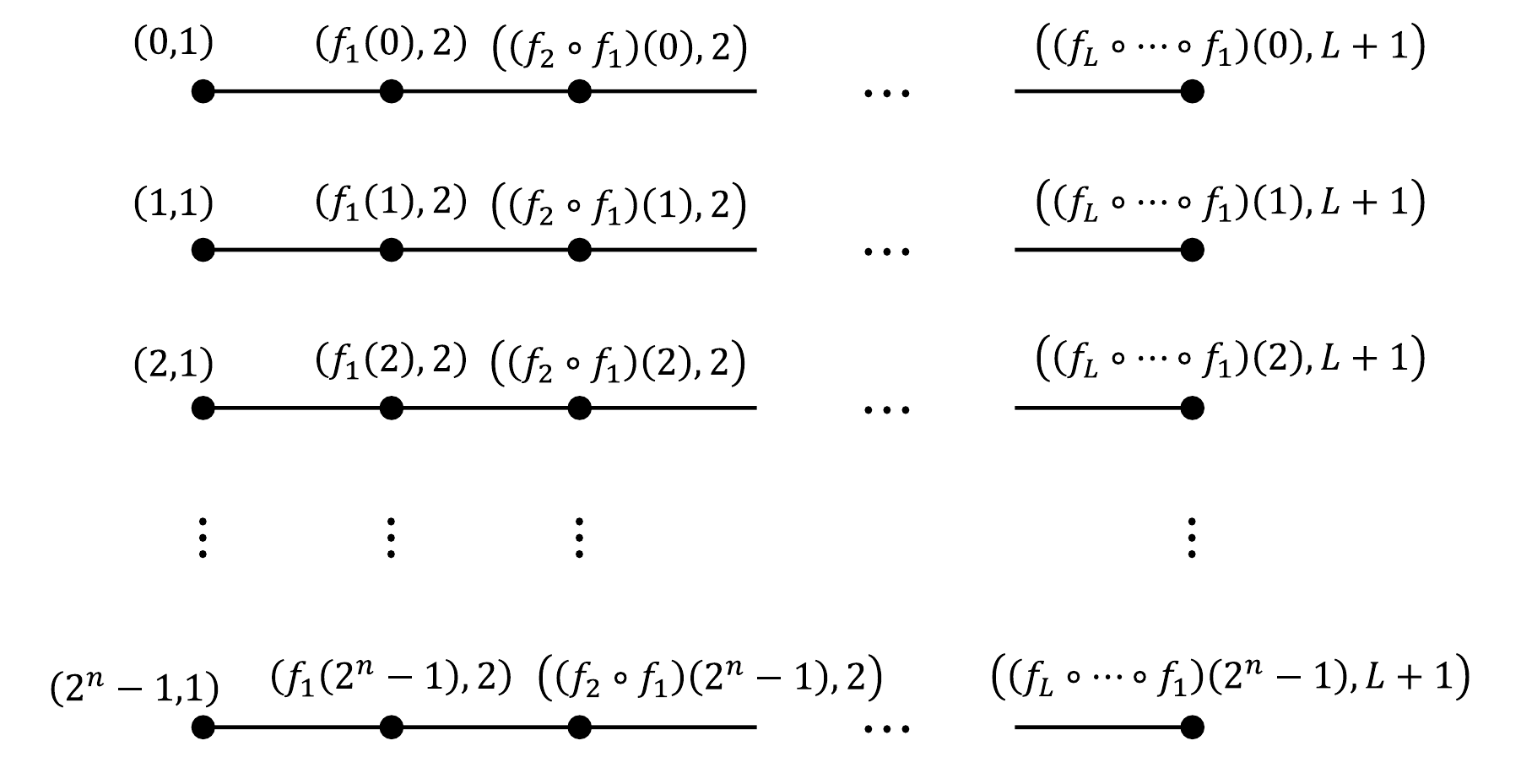}
}
\caption{The underlying two-dimensional (inherently one-dimensional) lattice of the reduced Hamiltonian in Eq.~\eqref{eq:FK-classical-reversible}.}
\label{fig:1d_system_classical_reversible_circuit}
\end{figure}

Next, we consider the time evolution of classical dynamics discussed in Sec.~\ref{subsec:classical}:
\begin{equation}
\dot{y}(t) + i \sqrt{A_{\mathrm{FK}}} y(t) = e^{i \sqrt{A_{\mathrm{FK}}} t} \left( \dot{y}(0) + i \sqrt{A_{\mathrm{FK}}} y(0) \right),
\end{equation}
with $y(0)=0$ and $\dot{y}(0) = \ket{1} \ket{\psi_0}$, and discuss how to obtain the sufficiently large overlap with the desired state for some time $t=O(L^2 \log{(L)})$.
The discussion below is well-studied (see e.g.~\cite{peres1985reversible, babbush2023exponential, stroeks2024solving}), so we provide a sketch of the proof.
The Feynman-Kitaev Hamiltonian can be expressed as the product of two operators
\begin{equation}
W\coloneqq \sum_{l=0}^L \ket{l+1}\bra{l+1} \otimes U_l \dots U_0,\quad
J\coloneqq 2I - \sum_{l=1}^L \left(  \ket{l+1}\bra{l} + \ket{l}\bra{l+1} \right) \otimes I,
\end{equation}
such that $H_{\mathrm{FK}}=WJW^\dag$.
Then the solution to Newton's equation is written as
\begin{equation}
\begin{split}
\dot{y}(t)
&=\mathrm{Re}\left( e^{i \sqrt{A_{\mathrm{FK}}} t} \ket{1} \ket{\psi_0} \right)
= \cos{\left( \sqrt{A_{\mathrm{FK}}} t \right)} \ket{1} \ket{\psi_0} \\
&\hspace{3cm} = W \cos{\left( \sqrt{J} t \right)} W^\dag \ket{1} \ket{\psi_0} 
= \sum_{l=1}^{L+1} \alpha_l(t) \ket{l} \otimes U_{l-1}\dots U_0\ket{\psi_0} ,
\end{split}
\end{equation}
where the coefficients $\alpha_l(t)$ is given by
\begin{equation}
\cos{\left( \sqrt{J} t \right)} \ket{1} = \sum_{l=1}^{L+1} \alpha_l(t) \ket{l}.
\end{equation}
Thus we only have to analyze the time evolution under $J$ to achieve the inverse polynomial overlap $|\alpha_{L+1}(t)|^2=\Omega(1/\mathrm{poly}(n))=\Omega(1/\mathrm{polylog}(N))$.
According to Ref.~\cite{babbush2023exponential}, there exists a probability distribution $p(t)$ such that 
\begin{equation}
\left| \sum_{t=0}^T p(t) |\alpha_{L+1}(t)|^2 - \frac{3}{4(L+2)} \right| \leq 2\varepsilon,
\end{equation}
for $T=O(L^2 \log{(1/\varepsilon)})$.
Choosing $\varepsilon = 1/(4(L+2))$, we obtain that $\sum_{t=0}^T p(t) |\alpha_{L+1}(t)|^2=\Omega (1/L)$ for $T=O(L^2 \log{(L)})$.
This implies that there must exist a $t=O(L^2\log{(L)})$ such that $|\alpha_{L+1}(t)|^2=\Omega (1/L)$.
Thus we can simulate the classical reversible circuit $U_{L}\dots U_0\ket{\psi_0}$ in $\mathrm{polylog}(N)$-time by simulating this classical system for all $t\in \{1,\dots ,T=\mathrm{polylog}(N)\}$.
It is easy to see that the Hamiltonian discussed above can be simulated by an $O(n)$-qubit quantum algorithm in $\mathrm{polylog}(N)$-time.
This finishes the proof.

Note that we can also establish the perfect overlap $|\alpha_{L+1}(t)|^2 = 1$ for $t=2\pi L =O(L)$ by adjusting the weight of $J$ as in Ref.~\cite{babbush2023exponential} (see also Ref.~\cite{peres1985reversible}).
This problem is known as perfect state transfer~\cite{christandl2004perfect, kay2010perfect, bosse2017coherent}, and we can use the ``mass–spring chains'' version of this result~\cite{vaia2020persymmetric, scherer2022analytic, lemay2016racah}.

Here we discuss one corollary of the above result.
Theorem~\ref{thm:classical-simulation-sampling-1D-GLI} indicates that a short-time evolved state $e^{-iHt}\ket{\psi}$ under the geometrically local Hamiltonian is classically simulatable in the \textit{weak} sense.
On the other hand, we can also observe that such a state cannot be simulated in the \textit{strong} sense (i.e., simulation of any marginal probability of the state) as a corollary of this Section.
This is followed by the fact shown in Ref.~\cite{nest2008classical} that strong simulation of HT circuit (first a layer of Hadamard gates which are applied to a subset of the qubits, followed by a round of classical reversible gates, e.g., $\{ X, \mathrm{CNOT}, \mathrm{Toffoli} \}$) is $\#\mathrm{P}$-hard.

\section{Computational complexity of simulating long-time dynamics under the geometrically local matrix}\label{sec:complexity-long-time-GLI}
In this section, we show that the computational complexity of simulating long-time ($\mathrm{poly}(N)$-time) dynamics under the geometrically local matrix is equivalent to that of $\mathrm{poly}(N)$-time and $O(n)$-space quantum computation.
In this work, we only consider the quantum algorithms that run in $\mathrm{poly}(t,n)$-time, and thus the long-time dynamics under the geometrically local matrix with system size of $N=2^n$ is simulated by quantum computers in $O(t)=\mathrm{poly}(N)$-time and $O(\log{(N)})=O(n)$-space.
Thus we prove the opposite direction in next Subsection~\ref{subsec:universality-long-time-2d-GLI}.

It is worth explaining why the dequantization technique for the short-time regime cannot be extended to the long-time regime.
For a geometrically local matrix, the time-evolution operator has an effective light cone, that is, information initially localized at a site affects only sites within a region of radius growing with the evolution time.
Our dequantized algorithm relies on this light-cone property.
In the short-time regime, this radius is at most $\mathrm{polylog}(N)$, and therefore the resulting dequantized algorithm has a complexity of $\mathrm{polylog}(N) = \mathrm{poly}(n)$.
In the long-time regime, however, the radius can be $\mathrm{poly}(N)$, so our dequantized algorithm may require $\mathrm{poly}(N)$-time and $\mathrm{poly}(N)$-space.
Thus this dequantization technique does not yield a classical simulation with resources comparable to $\mathrm{poly}(N)$-time and $\mathrm{poly}(n)$-space quantum algorithms.
Also note that this long-range information propagation in the long-time regime effectively induces long-range interaction, which enables us to encode quantum circuits into the long-time dynamics.

It is also worth discussing the implication of this result of universality.
From this result, we can say that simulating the long-time dynamics of such geometrically local systems has a super-polynomial quantum advantage unless any $\exp(n)$-time and $\mathrm{poly}(n)$-space quantum circuit is simulated by classical computation with the same resources.
This can be understood from two perspectives of time and space complexity.
First, if we are given a classical resource of only $\exp(n)$-time, classical computers can simulate the $\exp(n)$-time and $\mathrm{poly}(n)$-space quantum circuit using $\exp{(n)}$-space by the direct calculation of multiplication of matrices and vectors.
In this case, quantum computers obtain an exponential space advantage over classical computers.
Second, if we are given a classical resource of only $\mathrm{poly}(n)$-space, classical computers can simulate the $\exp(n)$-time and $\mathrm{poly}(n)$-space quantum circuit using $\exp{(n^2)}$-time.
This can be accomplished by the state-of-the-art classical algorithm in Theorem 4.1 in Ref.~\cite{aaronson2016complexity}.
Thus, in this case, quantum computers have a super-polynomial time advantage over classical computers.

\subsection{Universality of long-time dynamics under the 2D geometrically local Hamiltonian}\label{subsec:universality-long-time-2d-GLI}
Now we show that the computational complexity of simulating long-time dynamics is at least as hard as that of $\mathrm{poly}(N)$-time and $O(n)$-space quantum computation. 
To show this, in the following, we show that quantum circuits with $n$-qubit and $\mathrm{poly}(N)$-depth can be simulated by the long-time dynamics of coupled harmonic oscillators with geometrically local interactions.
Our construction is based on Ref.~\cite{babbush2023exponential}.

First, we construct the Hamiltonian as in Sec.~\ref{subsec:1DGLI-simulate-QC-BPgate}:
\begin{equation}\label{eq:FK-quantum-long}
A_{\mathrm{L}} \coloneqq 3I - \sum_{l=1}^L \left(  \ket{l+1}\bra{l} \otimes U_l + \ket{l}\bra{l+1} \otimes U_l^\dag \right),
\end{equation}
where $U = U_{L} \dots U_1$ is a $n$-qubit unitary operator generated by $L=\exp{(n)}=\mathrm{
poly}(N)$ quantum gates belonging to a universal gate set of $\{ H, \mathrm{CNOT}, \mathrm{Toffoli} \}$.
Without loss of generality, we assume that all Hadamard gates $H$ act on the last qubit to impose locality to the Hamiltonian $A_{\mathrm{L}}$.
We add CNOT gate which can generate SWAP gates so that we satisfy the above property of Hadamard gates, while only $\{ H, \mathrm{Toffoli} \}$ is a universal gate set~\cite{shi2002both, aharonov2003simple}.

The Hamiltonian of a classical system of coupled harmonic oscillators $A=\sqrt{M}^{-1} F \sqrt{M}^{-1}$ defined in Sec.~\ref{subsec:classical} has non-positive off-diagonal terms.
However, the Hamiltonian in Eq.~\eqref{eq:FK-quantum-long} has positive terms since the Hadamard gate has a negative term, so this Hamiltonian does not imply that classical system.
To address this issue, we add one ancilla qubit and use the dilated version of the Hadamard gate as used in Ref.~\cite{babbush2023exponential} (see also Refs.~\cite{jordan2010quantum, childs2014bose}):
\begin{equation}\label{eq:dilated-hadamard}
H_{\mathrm{dil}} \coloneqq \frac{1}{\sqrt{2}}
\begin{pmatrix}
   I_2 & I_2 \\
   I_2 & X
\end{pmatrix},
\end{equation}
where $I_2$ is the identity operator on the two-dimensional space and $X$ is the Pauli $X$ operator.
Note that $H_{\mathrm{dil}}$ is not a unitary operator.
However, we can observe that $H_{\mathrm{dil}}$ acts as the unitary operator $H$ within the subspace where the last ancilla qubit is $\ket{-}$ since
\begin{equation}
H_{\mathrm{dil}} \left( I_2\otimes \ket{-} \right) 
= \frac{1}{\sqrt{2}}
\begin{pmatrix}
   I_2 & I_2 \\
   I_2 & X
\end{pmatrix}
\left( I_2\otimes \ket{-} \right) 
= \frac{1}{\sqrt{2}}
\begin{pmatrix}
   I_2 & I_2 \\
   I_2 & - I_2
\end{pmatrix}
\left( I_2\otimes \ket{-} \right) 
= \left( H \otimes I_2 \right)  \left( I_2\otimes \ket{-} \right) .
\end{equation}
This concludes that the new Hamiltonian
\begin{equation}\label{eq:FK-quantum-long-dilated}
A_{\mathrm{L,dil}} \coloneqq 3I - \sum_{l=1}^L \left(  \ket{l+1}\bra{l} \otimes W_l + \ket{l}\bra{l+1} \otimes W_l^\dag \right),
\end{equation}
where $W_l=H_{\mathrm{dil}}$ if $U_l$ is the Hadamard gate and $W_l=U_l\otimes I_2$ otherwise, has the same operation as $A_{\mathrm{L}}$ on the subspace where the last qubit is $\ket{-}$.
Note that $A_{\mathrm{L,dil}}$ has only non-positive off-diagonal terms and is diagonally dominant since the absolute value of the sum of the off-diagonal elements of the fixed column is upper-bounded by $1+2/\sqrt{2}$.
Here, without loss of generality, we assume that our circuit does not have two consecutive Hadamard gates since Hadamard gates act on the same qubit.

Now we have another issue that the classical system which $A_{\mathrm{L,dil}}$ implies has non-local interactions.
As shown in Fig.~\ref{fig:2d_system_quantum_circuit_long}, this non-locality comes from CNOT and Toffoli gates while Hadamard gates can be represented by 3-local interactions.
To solve this, if $W_l$ is CNOT or Toffoli gate, we just decompose $W_l$ into the product of unitaries which can be represented by 2-local interactions.
As in Fig.~\ref{fig:2d_system_quantum_circuit_long_local},  $W_l$ can be decomposed into the product of at most $O(2^{2n})$ unitaries $W_l=V_{l,O(2^{2n})}\dots V_{l,1}$, where each unitary $V_{l,i}$ exchanges only two adjacent computational bases.
Using this decomposition, we obtain the final Hamiltonian:
\begin{equation}\label{eq:FK-quantum-long-dilated-local}
A_{\mathrm{L,local}} \coloneqq 3I - \sum_{m=1}^{L\times O(2^{2n})} \left(  \ket{m+1}\bra{m} \otimes V_m + \ket{m}\bra{m+1} \otimes V_m^\dag \right),
\end{equation}
which is $(3,\mathcal{N}(3))$-geometrically local Hamiltonian in two-dimension.

\begin{figure}
\centerline{
\includegraphics[width=120mm, page=1]{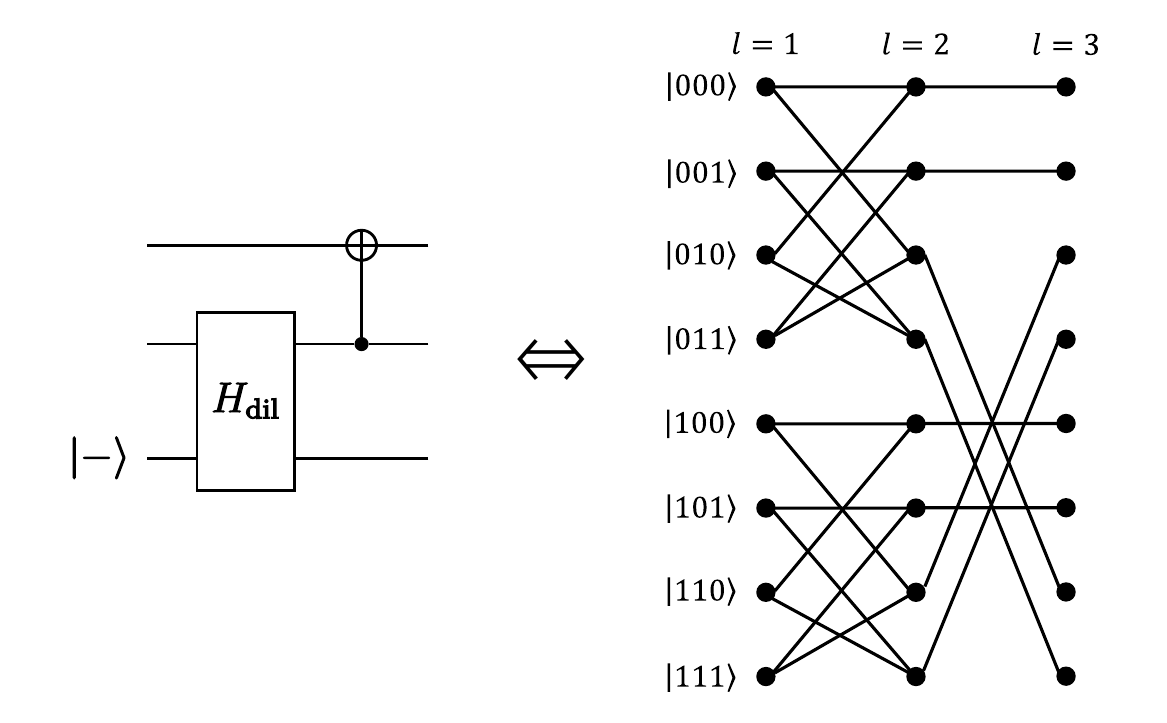}
}
\caption{An example of the underlying two-dimensional system of the reduced Hamiltonian in Eq.~\eqref{eq:FK-quantum-long-dilated}.}
\label{fig:2d_system_quantum_circuit_long}
\end{figure}

\begin{figure}
\centerline{
\includegraphics[width=120mm, page=1]{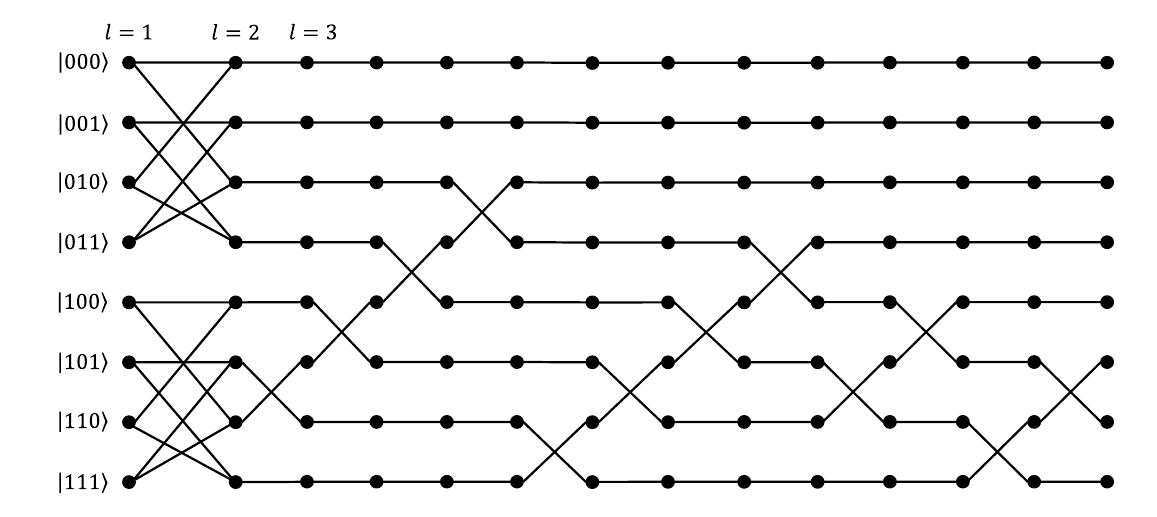}
}
\caption{Two-dimensional system with local interaction which has the same operation as the system with non-local interactions in Fig.~\ref{fig:2d_system_quantum_circuit_long}.}
\label{fig:2d_system_quantum_circuit_long_local}
\end{figure}

Finally, we consider the time evolution of the classical system discussed in Sec.~\ref{subsec:classical}:
\begin{equation}
\dot{y}(t) + i \sqrt{A_{\mathrm{L,local}}} y(t) = e^{i \sqrt{A_{\mathrm{L,local}}} t} \left( \dot{y}(0) + i \sqrt{A_{\mathrm{L,local}}} y(0) \right),
\end{equation}
with $y(0)=0$ and $\dot{y}(0) = \ket{1} \ket{\psi_0}\ket{-}$, and discuss how to obtain the overlap with the desired state for some time $t=O(L^2 2^{4n})$.
The discussion below is almost the same as that in Sec.~\ref{subsec:1DGLI-simulate-QC-BPgate} and Ref.~\cite{babbush2023exponential}.
Now the Hamiltonian $A_{\mathrm{L,local}}$ within the subspace where the last ancilla qubit is $\ket{-}$ can be written as
\begin{equation}\label{eq:FK-quantum-long-dilated-local-dash}
A_{\mathrm{L,local}}' \coloneqq 3I - \sum_{m=1}^{L\times O(2^{2n})} \left(  \ket{m+1}\bra{m} \otimes V_m' + \ket{m}\bra{m+1} \otimes V_m'^\dag \right),
\end{equation}
where $V_m'$ is $H\otimes I_2$ if $V_m=H_{\mathrm{dil}}$ and $V_m'=V_m$ otherwise.
Indeed, $V_m\neq H_{\mathrm{dil}}$ does not preserve this subspace.
However, recalling that $W_l=V_{l,O(2^{2n})}\dots V_{l,1}=U_l\otimes I_2$, we are in this subspace just before $V_m=H_{\mathrm{dil}}$ is applied.
All $V_m'$'s are unitary operators, which enables us to decompose the Hamiltonian $A_{\mathrm{L,local}}'$ into the product of two operators
\begin{equation}
W'\coloneqq \sum_{m=0}^{L\times O(2^{2n})} \ket{m+1}\bra{m+1} \otimes V_m' \dots V_0',
\end{equation}
\begin{equation}
J'\coloneqq 2I - \sum_{m=1}^{L\times O(2^{2n})} \left(  \ket{m+1}\bra{m} + \ket{m}\bra{m+1} \right) \otimes I,
\end{equation}
such that $A_{\mathrm{L,local}}'=W'J'W'^\dag$.
Putting it altogether, the solution to Newton's equation is written as
\begin{align}
\dot{y}(t)
&= \cos{\left( \sqrt{A_{\mathrm{L,local}}} t \right)} \ket{1} \ket{\psi_0} \ket{-}
= \cos{\left( \sqrt{A_{\mathrm{L,local}}'} t \right)} \ket{1} \ket{\psi_0} \ket{-} \\
&= W' \cos{\left( \sqrt{J'} t \right)} W'^\dag \ket{1} \ket{\psi_0} \ket{-}
= \sum_{m=1}^{L\times O(2^{2n})+1} \alpha_m(t) \ket{m} \otimes \left( V_{m-1}' \dots V_0' \ket{\psi_0} \right) \otimes \ket{-} ,
\end{align}
where the coefficients $\alpha_m(t)$ are given by
\begin{equation}
\cos{\left( \sqrt{J'} t \right)} \ket{1} = \sum_{m=1}^{L\times O(2^{2n})+1} \alpha_m(t) \ket{m}.
\end{equation}
According to the discussion in Sec.~\ref{subsec:1DGLI-simulate-QC-BPgate}, there must exist a $t=O(L^2 2^{4n} (n+\log{(L)}))=\mathrm{poly}(N)$ such that $|\alpha_{L\times O(2^{2n})+1}(t)|^2=\Omega (1/(L\times O(2^{2n}))) = \Omega (1/\mathrm{poly}(N))$.
Thus we can simulate the quantum circuit $U_{L}\dots U_0\ket{\psi_0} = V_{L\times O(2^{2n})-1}' \dots V_0' \ket{\psi_0}$ in $\mathrm{poly}(N)$-time by simulating this classical system for all $t\in \{1,\dots ,T=\mathrm{poly}(N)\}$.
This indicates that the $O(n)$-qubit quantum algorithm for long-time dynamics can simulate $\mathrm{poly}(N)$-time and $O(n)$-space quantum computation.

\section{Conclusion}\label{sec:conclusion}
We have studied the computational complexity of simulating classical linear dynamics of system size $N=2^n$ with geometrically local interactions on quantum computers. 
The contributions of this work can be summarized in three main aspects.
First, we dequantize the quantum algorithm for simulating short-time dynamics ($t=\mathrm{polylog}(N)$) of such systems. 
We achieved this by dequantizing the quantum eigenvalue transformation (QEVT) for geometrically local matrices. 
This implies that no exponential quantum advantage exists in this regime, even when both the initial state and the observable are global. 
We further proved that sampling from the time evolved states is also classically tractable. 
Second, we showed that simulating short-time dynamics is at least as hard as $\mathrm{polylog}(N)$-time and $O(n)$-space probabilistic classical computation by embedding the latter into the former.
Third, we extended our analysis to the long-time dynamics and showed that simulating such dynamics is equivalent in computational complexity to $\mathrm{poly}(N)$-time and $O(n)$-space quantum computation. 
This indicates a super-polynomial quantum advantage unless such quantum computations can be classically simulated with the same resources.
Our finding provides a new technical tool to understand the complexity of classical dynamics governed by partial differential equations and clarifies the computational boundary between classical and quantum systems.

Here, we discuss further the implications of our dequantized algorithm.
In the short-time regime, we have fully dequantized the quantum algorithm, and obtained classical algorithm achieves the same computational complexity as the quantum one up to a polynomial overhead.
Moreover, compared to existing classical algorithms such as forward and backward Euler methods, our algorithm requires exponentially smaller time and space complexities.
In the long-time regime, the dequantized algorithm requires $\exp{(n)}$-space.
Thus the $O(n)$-qubit quantum algorithm retains an exponential advantage in space complexity.
Also, our dequantized algorithm offers no particular advantage over existing classical methods in this regime.
Developing \textit{quantum-inspired classical} algorithms that outperform conventional methods for the long-time dynamics, which is commonly encountered in practical applications, remains a significant research challenge.

We now have several interesting future directions about our dequantized algorithm.
One important direction is the extension of our framework to nonlinear partial differential equations, which frequently appear in real-world physical systems, e.g. Navier–Stokes equations.
Understanding whether similar dequantization techniques can be applied to nonlinear dynamics remains a significant open problem.
Another avenue is to explore new applications of our dequantized algorithm. 
While the geometrically local matrices we introduced are motivated by classical physical systems, it is a valuable question whether our dequantization techniques can also be applied to problem settings in computer science or mathematics that are not directly derived from physical models.
Further improvements of the dequantized algorithms themselves in terms of runtime and space complexities may broaden their applicability and enhance their practical utility.
In the long-time case, our algorithm does not have an advantage over the existing classical algorithm as discussed above.
Thus improving this algorithm and showing the optimality may offer its practicality as a quantum-inspired classical algorithm.

We also have notable open questions about our other results.
For long-time dynamics, an important open question is whether stronger classical hardness results can be established.
In Sec.~\ref{subsec:universality-long-time-2d-GLI}, we embed quantum circuits into 2D geometrically local systems.
However, the resulting graph after embedding is not planar, which is not the case for wave equations.
Thus it is essential to show the universality of geometrically local systems on planar graphs or on 2D square lattices.
One potential approach to this problem is using the method established in the context of universal quantum computation on quantum walk via scattering theory~\cite{childs2009universal, blumer2011single, childs2013universal, childs2014bose}.
Extending this method to the system of coupled harmonic oscillators on planar graphs or on 2D square lattices is an interesting future direction.
Finally, more intensive works that investigate polynomial quantum speedups in simulating systems governed by geometrically local matrices are required.
As shown in Ref.~\cite{babbush2021focus}, quantum speedups greater than quartic speedups may have a practical advantage.
Thus we need to analyze how much polynomial speedups can be extracted from this kind of problem.

\begin{acknowledgments}
We thank Kosuke Mitarai for useful discussions and for reading through a draft of this document. 
This work is supported by MEXT Q-LEAP Grant
No. JPMXS0120319794, JST CREST JPMJCR24I3, and JST COI-NEXT No.
JPMJPF2014.
K.S. is supported by JST SPRING Grant No. JPMJSP2138 and the $\Sigma$ Doctoral Futures Research Grant Program from The University of Osaka.
\end{acknowledgments}

\section*{Author contribution}
K.S. and K.F. contributed to the conceptual formulation of this work.
K.S. contributed to the design and analysis of the dequantized algorithm and the complexity analysis of the classical dynamics. 
All authors contributed to discussion of the results and
text writing.
Authors used generative AI tools to refine the text.


\bibliographystyle{quantum}
\bibliography{ref}


\appendix

\clearpage

\section{Other Applications of the dequantized algorithm for the geometrically local matrix}\label{appsec:other-applications}

\subsection{Application to the simulation of wave equations}\label{appsubsec:wave}
The wave equation is written as 
\begin{equation}
\frac{\mathrm{d}^2}{\mathrm{d}t^2} \phi(t,\bm{x}) = c^2 \nabla^2 \phi(t,\bm{x}) ,
\end{equation}
where $c$ is a wave propagation speed.
After the discretization using the central difference scheme, we obtain
\begin{equation}
\frac{\mathrm{d}^2}{\mathrm{d}t^2} \phi(t) = - c^2 \frac{1}{a^2} L \phi(t) ,
\end{equation}
where $L$ is the Laplacian operator and $a$ is a lattice spacing.
Here we note that $L$ is an $(1,\mathcal{N}(1))$-geometrically local Hamiltonian, where $\mathcal{N}(r) = O(r^D)$.
Comparing the above equation to the Newton's equation $\ddot{y}(t) = -A y(t)$, where $A=\sqrt{M}^{-1} F \sqrt{M}^{-1} \geq 0$, we can observe that the wave equation is a special case of classical systems of coupled harmonic oscillators with geometrically local interactions in Sec.~\ref{subsec:classical}.
That is, substituting $m_i=1$ and $\kappa_{ii}=\frac{c^2}{a^2}\sum_j L_{ij}$ for all $i$ and $\kappa_{ij} = -\frac{c^2}{a^2} L_{ij}$ for $i\neq j$ to the Newton's equation, we obtain the wave equation.
Using this substitution, as in Sec.~\ref{subsec:classical}, we can transform the wave equation to the Hamiltonian dynamics as
\begin{equation}
\frac{\mathrm{d}}{\mathrm{d}t} \ket{\Phi(t)} = iH_{\mathrm{wave}} \ket{\Phi(t)},
\end{equation}
where 
\begin{equation}\label{eq:dilated-hamiltonian-wave}
H_{\mathrm{wave}} \coloneqq 
\begin{pmatrix}
   0 & B_{\mathrm{wave}} \\
   B_{\mathrm{wave}}^\dag & 0
\end{pmatrix},
\end{equation}
$B_{\mathrm{wave}}$ is a $N\times N_\mathrm{col}$ matrix such that $B_{\mathrm{wave}} B_{\mathrm{wave}}^\dag = \frac{c^2}{a^2} L$,
\begin{equation}
\ket{\Phi(t)} \coloneqq
\frac{1}{\sqrt{2E}}
\begin{pmatrix}
   \dot{\phi}(t) \\
   i B_{\mathrm{wave}}^\dag \phi(t)
\end{pmatrix},
\end{equation}
and $\sqrt{2E}$ is a normalization factor.
Also we can choose $B_{\mathrm{wave}}$ as
\begin{equation}\label{eq:B-dag-wave}
B_{\mathrm{wave}}^\dag \ket{i} = \sum_{i}\sqrt{\frac{c^2}{a^2}\sum_j L_{ij}} \ket{i}\ket{i} + \sum_{j> i}\sqrt{-\frac{c^2}{a^2} L_{ij}} \ket{i}\ket{j} - \sum_{j< i}\sqrt{-\frac{c^2}{a^2} L_{ij}} \ket{j}\ket{i}.
\end{equation}

The following corollary is the direct consequence of Theorem~\ref{thm:classical-time}, which provides the complexity of the dequantized algorithm for simulating the wave equation.

\begin{cor}
Let $t\in \mathbb{R}$ be an evolution time, $a$ be a lattice spacing and $\varepsilon\in(0,1]$ be a precision.
Suppose that we are given
\begin{itemize}
\item query-access to $\dot{\phi}(0)\in \mathbb{R}^N$,
\item query-access to $\phi(0)\in \mathbb{R}^N$,
\item query-access to  Laplacian $L\in \mathbb{R}^{N\times N}$ which is obtained by discretizing the wave equation,
\item total energy: $E= \frac{1}{2} \sum_{i} \dot{\phi}_i(0)^2 + \frac{1}{2} \frac{c^2}{a^2} \sum_{i} (\sum_j L_{ij}) \phi_i(0)^2 - \frac{1}{2} \frac{c^2}{a^2} \sum_{j> i}L_{ij} (\phi_i(0)-\phi_j(0))^2$,
\item $\zeta$-sampling-and-query-access to $v\in \mathbb{C}^{N+N(N+1)/2}$.
\end{itemize}
Then, for $\zeta\leq \varepsilon/18$, we can calculate
\begin{equation}
v^\dag \cdot \Phi(t)
=
v^\dag \cdot 
\frac{1}{\sqrt{2E}}
\begin{pmatrix}
   \dot{\phi}(t) \\
   i B_{\mathrm{wave}}^\dag \phi(t)
\end{pmatrix},
\end{equation}
classically with a probability at least $1-\delta$ and an additive error $\varepsilon$.
The runtime of the classical algorithm is
\begin{equation}
O\left( \frac{\log{(1/\delta)}}{\varepsilon^2}\left( d_{\mathrm{exp}}^2 \mathcal{N}(d_{\mathrm{exp}} )\mathcal{N}(1) \bm{q}(L) + d_{\mathrm{exp}} \mathcal{N}(d_{\mathrm{exp}} ) \left( \bm{q}(\dot{\phi}(0)) + \bm{q}(\phi(0)) \right) +\bm{sq}(v)\right) \right),
\end{equation}
where $d_{\mathrm{exp}} \coloneqq O( (ct/a) \sqrt{\mathcal{N}(1)} + \log(1/\varepsilon))$, or emphasizing the dependence on $t$, $a$, $\varepsilon$ and $\delta$, given the spatial dimension $D$,
\begin{equation}
\tilde{O}\left( \frac{(t/a)^{2+D}\log{(1/\delta)}}{\varepsilon^2} \right).
\end{equation}
\end{cor}

\subsection{Application to the simulation of advection equations}\label{appsubsec:advection}
The dequantized algorithm for simulating advection equations is much simpler than that discussed above.
We consider the advection equation in $D$-dimension with the constant velocity $\bm{v}$:
\begin{equation}\label{eq:advection}
\frac{\partial u(t, \bm{x})}{\partial t}+\bm{v} \cdot \nabla u(t, \bm{x})=0,
\end{equation}
where $u$ is a scalar field, $t$ is a time and $\bm{x}$ is a spatial coordinate.
As in Ref.~\cite{sato2024hamiltonian}, we can rewrite Eq.~\eqref{eq:advection} as
\begin{equation}\label{eq:advection2}
\frac{\partial u(t, \bm{x})}{\partial t}=-i(-i\bm{v} \cdot \nabla) u(t, \bm{x}),
\end{equation}
and then discretize the field $u$ and the Hamiltonian $-i\bm{v} \cdot \nabla$ as
\begin{equation}\label{eq:advection-field}
\ket{u(t)} \coloneqq \sum_{j_1=1}^{N}\dots \sum_{j_D=1}^{N} u(t,x_{j_1},\dots ,x_{j_D})\ket{j_1}\dots \ket{j_D},
\end{equation}
\begin{equation}\label{eq:advection-ham}
H= -iv_1 D^{\pm}\otimes I^{\otimes D-1} -iv_2 I\otimes D^{\pm} \otimes I^{\otimes D-2} - \dots -iv_D I^{\otimes D-1} \otimes D^\pm,
\end{equation}
where $N=2^n$, $(D^\pm \sum_{j=1}^{N}u(x_{j})\ket{j})_i = \frac{(u(x_{i+1})-u(x_{i-1}))\ket{i}}{2a}$ and $a$ is the lattice spacing.
It is worth noting that our dequantized algorithm can be applied to the case of spatially varying parameters and various boundary conditions as long as $(-i\bm{v} \cdot \nabla)$ is self-adjoint.
Now we set each site on $D$-dimensional cubic lattice with lattice spacing $1$ (: constant), which must be distinguished with that for discretization $a$.
Then we can observe that the Hamiltonian $H$ in Eq.~\eqref{eq:advection-ham} is a $(1,\mathcal{N}(1))$-geometrically local matrix, where $\mathcal{N}(r) \coloneqq \max_{i=\{1,\dots,N\}} \left| \left\{ j | d(i,j) \leq r \right\} \right| = O(r^D)$.

As a result of Theorem~\ref{thm:dequantize-qevt}, we obtain the following statement.
\begin{cor}\label{cor:deqantized-alg-advection}
Let $t\in \mathbb{R}$ be an evolution time, $D=O(1)$ be a dimension, $\varepsilon\in(0,1]$ be a precision and $v_{\mathrm{max}}=\max_{i\in\{1,\dots,D\}}|v_i|$ be the maximum of velocities.
Suppose that we are given
\begin{itemize}
\item query-access to $u(0) \coloneqq \sum_{j_1=1}^{N}\dots \sum_{j_D=1}^{N} u(0,x_{j_1},\dots ,x_{j_D})e_{j_1}\otimes \dots \otimes e_{j_D} \in\mathbb{R}^{N^D}$ such that $\|u(0)\|\leq 1$.
\item query-access to $H\in \mathbb{R}^{N^D\times N^D}$ which is defined in Eq.~\eqref{eq:advection-ham},
\item $\zeta$-sampling-and-query-access to $v\in \mathbb{R}^{N^D}$ such that $\|v\|\leq 1$.
\end{itemize}
Then, for $\zeta\leq \varepsilon/18$, we can calculate $v^\dag u(t)$ classically with a probability at least $1-\delta$ and an additive error $\varepsilon$.
The runtime of the classical algorithm is 
\begin{equation}
O\left( \frac{\log{(1/\delta)}}{\varepsilon^2}\left( d_{\mathrm{exp}}^2 \mathcal{N}(d_{\mathrm{exp}})\mathcal{N}(1) \bm{q}(H) + d_{\mathrm{exp}} \mathcal{N}(d_{\mathrm{exp}}) \bm{q}(u(0)) +\bm{sq}(v)\right) \right) ,
\end{equation}
where $d_{\mathrm{exp}}\coloneqq O(Dv_{\mathrm{max}}t/a + \log(1/\varepsilon))$, or emphasizing the dependence on $t$, $v_{\mathrm{max}}$, $a$, $\varepsilon$ and $\delta$, given the spatial dimension $D$, 
\begin{equation}
\tilde{O}\left( \frac{(Dv_{\mathrm{max}}t/a)^{2+D} \log{(1/\delta)}}{\varepsilon^2} \right) .
\end{equation}
\end{cor}
\begin{proof}
For some $\alpha\in \mathbb{R}$, let $P_{\mathrm{exp}}\in\mathbb{C}[x]$ be the polynomial approximation of exponential function $e^{-itx}$ on the domain $[-\alpha,\alpha]$ such that
\begin{equation}
\left\| P_{\mathrm{exp}}(x) - e^{-itx} \right\|_{[-\alpha,\alpha]} \leq \varepsilon, \quad \left\| P_{\mathrm{exp}}(x) \right\|_{[-\alpha,\alpha]} \leq 1.
\end{equation}
According to Ref.~\cite{gilyen2019quantum}, there exists such an polynomial approximation of degree $O(\alpha t + \log(1/\varepsilon))$.

Recalling that $v^\dag u(t)=v^\dag e^{-iHt}u(0)$, we only have to construct a $\varepsilon/2$-approximation $P_{\mathrm{exp}}$ of $e^{-itx}$ on the domain $[-\|H\|,\|H\|]$ and estimate $v^\dag P_{\mathrm{exp}}(H) u(0)$ with an additive error $\varepsilon/2$.
$H$ is $2D$-sparse and its max norm $\|H\|_{\mathrm{max}}=v_{\mathrm{max}}/2a$.
This implies $\|H\|\leq 2D\cdot \|H\|_{\mathrm{max}} = Dv_{\mathrm{max}}/a$ and we obtain $P_{\mathrm{exp}}$ of degree $d_{\mathrm{exp}}\coloneqq O(Dv_{\mathrm{max}}t/a + \log(1/\varepsilon))$.
Also note that $\| P_{\mathrm{exp}}(H) \| \leq 1$.

As discussed earlier, $H$ is a $(1,\mathcal{N}(1))$-geometrically local matrix and $\mathcal{N}(r) = O(r^D)$.
Combined with Theorem~\ref{thm:dequantize-qevt}, we obtain a classical algorithm which calculate $v^\dag P_{\mathrm{exp}}(H) u(0)$ with an additive error $\varepsilon/2$.
For $\zeta\leq \varepsilon/18$, the runtime of this classical algorithm is given by
\begin{equation}
O\left( \frac{\log{(1/\delta)}}{\varepsilon^2}\left( d_{\mathrm{exp}}^2 \mathcal{N}(d_{\mathrm{exp}})\mathcal{N}(1) \bm{q}(H) + d_{\mathrm{exp}} \mathcal{N}(d_{\mathrm{exp}}) \bm{q}(u(0)) +\bm{sq}(v)\right) \right) .
\end{equation}
Emphasizing the dependence on $t$, $v_{\mathrm{max}}$, $a$, $\varepsilon$ and $\delta$, we can write the complexity as
\begin{equation}
\tilde{O}\left( \frac{(Dv_{\mathrm{max}}t/a)^{2+D} \log{(1/\delta)}}{\varepsilon^2} \right) .
\end{equation}
\end{proof}

\subsection{Application to the simulation of single particle Schr\"{o}dinger equations}\label{appsubsec:schrodinger}
The single particle Schr\"{o}dinger equation is given by
\begin{equation}
i \frac{\mathrm{d}}{\mathrm{d}t}\psi(t, \bm{x}) = (-\nabla^2 + V(\bm{x})) \psi(t, \bm{x}),
\end{equation}
where $V: \mathbb{R}^D \rightarrow \mathbb{R}$ is a potential function and $D$ is the spatial dimension.
After the discretization using the central difference scheme, we get
\begin{equation}
i \frac{\mathrm{d}}{\mathrm{d}t}\ket{\psi(t)} = (\frac{1}{a^2}L + V) \ket{\psi(t)},
\end{equation}
where $L$ is the Laplacian, $a$ is a lattice spacing and $V$ is a diagonal matrix.
It is easy to see that $H_{\mathrm{SE}}\coloneqq \frac{1}{a^2}L + V$ is an $(1,\mathcal{N}(1))$-geometrically local matrix, where $\mathcal{N}(r)=O(r^D)$.
Its norm is upper-bounded by $\| H_{\mathrm{SE}} \| \leq \frac{1}{a^2} \|L\| + \|V\| \leq \frac{1}{a^2}(2D+1)2D + V_{\mathrm{max}} = O(\frac{D^2}{a^2} + V_{\mathrm{max}})$, where $V_{\mathrm{max}}\coloneqq \|V\|_{\mathrm{max}}$.

Then we obtain the dequantized algorithm for simulating the above dynamics as follows.

\begin{cor}\label{cor:dequantized-alg-schrodinger}
Let $t\in \mathbb{R}$ be an evolution time, $a$ be a lattice spacing and $\varepsilon\in(0,1]$ be a precision.
Suppose that we are given
\begin{itemize}
\item query-access to $\ket{\psi(0)}\in \mathbb{C}^N$,
\item query-access to $V \in \mathbb{R}^N$, where $V_i$ is the potential of $i$-th site for all $i\in\{1,\dots,N\}$,
\item query-access to Laplacian $L\in \mathbb{R}^{N\times N}$,
\item $\zeta$-sampling-and-query-access to $\ket{v}\in \mathbb{C}^{N}$.
\end{itemize}
Then, for $\zeta\leq \varepsilon/18$, we can calculate $\braket{v|\psi(t)}$ classically with a probability at least $1-\delta$ and an additive error $\varepsilon$.
The runtime of the classical algorithm is
\begin{equation}
O\left( \frac{\log{(1/\delta)}}{\varepsilon^2}\left( d_{\mathrm{exp}}^2 \mathcal{N}(d_{\mathrm{exp}} )\mathcal{N}(1) \left( \bm{q}(L) + \bm{q}(V) \right) + d_{\mathrm{exp}} \mathcal{N}(d_{\mathrm{exp}} ) \bm{q}(\psi(0))  +\bm{sq}(v)\right) \right),
\end{equation}
where $d_{\mathrm{exp}} \coloneqq O( t (D^2/a^2 + V_{\mathrm{max}}) + \log(1/\varepsilon))$, or emphasizing the dependence on $t$, $a$, $V_{\mathrm{max}}$, $\varepsilon$ and $\delta$, given the spatial dimension $D$,
\begin{equation}
\tilde{O}\left( \frac{(t(D^2/a^2+V_{\mathrm{max}}))^{2+D}\log{(1/\delta)}}{\varepsilon^2} \right).
\end{equation}
\end{cor}
\begin{proof}
As in the proof of Corollary~\ref{cor:deqantized-alg-advection}, for the Hamiltonian $H_{\mathrm{SE}}$, there exists a polynomial $P_{\mathrm{exp}}\in \mathbb{C}[x]$ of degree $O(\|H_{\mathrm{SE}}\| t + \log{(1/\varepsilon)})$ such that $\| P_{\mathrm{exp}}(H_{\mathrm{SE}}) - e^{-iH_{\mathrm{SE}}t} \| \leq \varepsilon$ and $\| P_{\mathrm{exp}}(H_{\mathrm{SE}}) \| \leq 1$.
Now $\| H_{\mathrm{SE}} \|$ is upper-bounded by $O(\frac{D^2}{a^2} + V_{\mathrm{max}})$, which provides the $\varepsilon/2$-approximation $P_{\mathrm{exp}}(H_{\mathrm{SE}})$ of $e^{-iH_{\mathrm{SE}}t}$ with degree $d_{\mathrm{exp}} \coloneqq O(t (\frac{D^2}{a^2} + V_{\mathrm{max}}) + \log{(1/\varepsilon)})$.

All that is left is estimating $\braket{v|P_{\mathrm{exp}}(H_{\mathrm{SE}})|\psi(0)}$ with an additive error $\varepsilon/2$.
Theorem~\ref{thm:dequantize-qevt} provides a classical algorithm which estimates $\braket{v|P_{\mathrm{exp}}(H_{\mathrm{SE}})|\psi(0)}$ with a probability at least $1-\delta$.
For $\zeta\leq \varepsilon/18$, the runtime of this classical algorithm is given by
\begin{equation}
O\left( \frac{\log{(1/\delta)}}{\varepsilon^2}\left( d_{\mathrm{exp}}^2 \mathcal{N}(d_{\mathrm{exp}})\mathcal{N}(1) \left(\bm{q}(L) + \bm{q}(V)\right) + d_{\mathrm{exp}} \mathcal{N}(d_{\mathrm{exp}}) \bm{q}(\psi(0)) +\bm{sq}(v)\right) \right) .
\end{equation}
Emphasizing the dependence on $t$, $V_{\mathrm{max}}$, $a$, $\varepsilon$ and $\delta$, given the spatial dimension $D$, we can write the complexity as
\begin{equation}
\tilde{O}\left( \frac{(t(D^2/a^2+V_{\mathrm{max}}))^{2+D}\log{(1/\delta)}}{\varepsilon^2} \right).
\end{equation}
\end{proof}

Here we discuss this result from the computational complexity aspect.
Although the system is a single-particle system in a fixed spatial dimension, the discretized wave function over $N$ spatial sites is an $N$-dimensional vector and nontrivial spatial structure can be introduced through the potential and boundary conditions.
Therefore, it is still meaningful to ask whether the quantum solver using $n=\log{N}$ qubits yields a super-polynomial advantage in $n$.
Corollary~\ref{cor:dequantized-alg-schrodinger} indicates that simulating the single particle Schr\"{o}dinger equation (also multi-particle with no interactions between particles) does not yield super-polynomial quantum speedups.
This fact is contrast to the recent result~\cite{zheng2024computational} showing that simulating the multi-particle Schr\"{o}dinger equation with two-body interactions is $\mathrm{BQP}$-complete, which implies an exponential quantum speedup in terms of the particle number.
Thus our work provides a complementary result that the interactions between two particles are essential for exponential speedups in simulating the  Schr\"{o}dinger operators.

\end{document}